\documentclass[final, journal]{IEEEtran}

\usepackage{amsmath, amssymb, amsfonts, amsxtra, amscd}
\usepackage{amsthm}
\usepackage{latexsym}
\usepackage{graphicx}
\usepackage{verbatim}
\usepackage{mathrsfs}
\usepackage{algorithmic}
\usepackage{float}
\usepackage{url}
\usepackage[lofdepth, lotdepth]{subfig}
\usepackage[top=2.5cm, bottom=3cm, left=2.5cm, right=2.5cm]{geometry}

\usepackage[usenames,dvipsnames]{pstricks}
\usepackage{epsfig}
\usepackage{pst-grad} 

\theoremstyle{plain}
  \newtheorem{theorem}{Theorem}[section]
  \newtheorem{lemma}[theorem]{Lemma}
  
  \newtheorem{corollary}[theorem]{Corollary}
\theoremstyle{definition}
  \newtheorem{definition}[theorem]{Definition}
  
  \newtheorem{example}[theorem]{Example}
  \newtheorem{remark}[theorem]{Remark}

\newcommand{\A}{{\mathcal A}}
\newcommand{\B}{{\mathcal B}}
\newcommand{\C}{{\mathcal C}}

\newcommand{\D}{{\mathcal D}}

\newcommand{\G}{{\mathcal G}}

\newcommand{\X}{{\mathcal X}}

\newcommand{\lam}{\lambda}

\newcommand{\nin}{\noindent}

\newcommand{\et}{{\emph{et al.}}}

\title{Parity Declustering for Fault-Tolerant Storage Systems via $t$-designs
\thanks{%
S. H. Dau is with the SUTD-MIT International Design Centre, Singapore University of Technology and Design (e-mail: sonhoang\_dau@sutd.edu.sg).  
This work was done when he was with the Division of Mathematical Sciences, School of Physical and Mathematical Sciences,
Nanyang Technological University.
}
\thanks{%
Y. Jia, C. Jin, W. Xi, and K. S. Chan are with the Data Storage
Institute (DSI), Agency For Science, Technology And Research (A*STAR), North Connexis Tower, Fusionopolis, Singpore 138632 (e-mails: \{jia\_yan, jin\_chao, xi\_weiya, chan\_kheong\_sann\}@dsi.a-star.edu.sg).    
} 
}
\author{Son Hoang Dau, Yan Jia, Chao Jin, Weiya Xi, Kheong Sann Chan}

\begin{document}
\maketitle

\begin{abstract}
Parity declustering allows faster reconstruction of a disk array when some disk fails. 
Moreover, it guarantees uniform reconstruction workload on all surviving disks. 
It has been shown that parity declustering for 
one-failure tolerant array codes can be obtained via Balanced Incomplete Block Designs.
We extend this technique for array codes that can tolerate an arbitrary number of 
disk failures via $t$-designs. 
\end{abstract}

\section{Introduction}

RAID (Redundant Array of Independent Disks) has been widely used as a large-scaled and 
reliable storage system since its introduction in 1988 \cite{PattersonGibsonKatz1988}. However, the key limitation of the 
first 6 levels of RAID (RAID-0 to RAID-5) is that system recovery can be possible with at most one disk failure. 
RAID-6 has been proposed as a new RAID standard, which requires that any one or two disk
failures can be fixed. Several types of codes that can correct two erasures have been proposed,  
such as Reed-Solomon (RS) code \cite{Plank1997}, EVEN-ODD code \cite{BlaumBradyBruckMenon1995}, B-code \cite{XuBohossianBruckWagner1998}, X-code \cite{XuBruck2008}, 
and RDP code \cite{Corbett2004}. Codes that allow the recovery from more than two failures have also been 
investigated \cite{FengDengBaoShen2005, FengDengBaoShen2005_2, HuangXu2008}. The main limitation of RS codes is the high encoding and decoding complexity, 
which involves computation over finite fields. The other types of codes, called array codes, are preferred by storage system designers due to the fact that their encoding and decoding requires
only XOR operations.

The majority of known array codes are MDS (Maximum Distance Separable) codes (see~\cite{MW_and_S}). MDS array codes
have optimal redundancy ($\delta$ redundant disks are used in a $\delta$-failure tolerant array code). The main
issue with them is that when $\delta$ disks fail, all data in every surviving disk has to be read for reconstruction.
This results in slow reconstruction time when disk capacities get larger and increases the possibility of 
another failure, which renders the reconstruction impossible. Moreover, as all disks must be fully accessed
for the recovery purpose, the system operates in its degraded mode: responses to user requests take 
longer time than usual. 

Parity declustering (or clustered RAID) was proposed by Muntz and Lui \cite{MuntzLui1990} as a data layout technique
that allows faster reconstruction and uniform reconstruction workloads 
on surviving devices during reconstruction of one disk failure.
Here, the reconstruction workload refers to the amount of data 
that needs to be accessed on the surviving disks in order to reconstruct the data on the failed disk.  
Faster reconstruction stems from the feature of the declustered-parity data layout
that requires only a \emph{partial} access instead of a full access to each surviving disk. 
In other words, the special layout allows 
reconstruction of data on a failed disk \emph{without} reading all data in every 
surviving disk. 
Muntz and Lui suggested that designing such 
a layout is a combinatorial block design problem, but gave no further details.
Holland and Gibson \cite{HollandGibson1992}, Ng and Mattson \cite{NgMattson1992} investigated 
the construction of parity-declustered data layouts from Balanced Incomplete Block Designs (BIBD).
The work of Reddy and Banerjee \cite{ReddyBannerjee1991} also followed the same approach, even
though they focused more on a special type of BIBDs. 

For codes that can tolerate $\delta \geq 2$ disk failures, it is also desirable to have a declustered-parity 
data layout. More specifically, we want to design a layout such that when at most $\delta$ disks fail, 
only \emph{a portion} of the disk content on each healthy disk needs to be accessed for the recovery process. 
Moreover, the reconstruction workload is distributed \emph{uniformly} to all surviving disks. 
There has been several work where parity declustering for $\delta$-failure tolerant codes ($\delta \geq 2$) are 
considered, such as \cite{AlvarezBurkhardCristian1997} and
\cite{Alvarez-et-1998}. However, none of them guarantee the uniform workloads during the reconstruction
of more than one disk. Corbett \cite{Corbett2008} proposed that two array codes of the
same size can be combined into a larger array that has \emph{almost} uniform reconstruction workloads when one or two disks fail. However, Corbett's method only achieves uniform workloads among the \emph{data} disks, 
not over all surviving disks (data disks and parity disks). Moreover, his construction produces an array 
code of a prohibitively large size, which is at least $\binom{n}{n/2} \times (n+4)$, where $n$ is the number of data disks in the final array code.   

We investigate the construction of declustered-parity layouts for codes that tolerate $t-1$ disk failures 
via $t$-designs ($t \geq 2$). In fact, BIBDs, which are used to decluster parities for one-failure tolerant codes, are
$2$-designs. The main idea is to start with an array code of $k$ columns that has uniform workloads for 
reconstruction of every $s \leq t -1$ columns. Then, the $k$ columns of this code are spread out over 
$n > k$ disks, using blocks of a $t$-$(n,k,\lam)$ design (see Section~\ref{sec:preliminaries} for all definitions). As a result, 
we obtain an array code with $n$ disks that possesses the following properties. Firstly, in order to recover
any $s \leq t - 1$ disks, only a portion of the disk
content, which is a designed parameter, must be read for disk recovery. Secondly,
the reconstruction workload is uniformly distributed to every surviving disk. And lastly, 
the parity units are distributed evenly over all disks, which eliminates hot spots during data update.
To the best of our knowledge, this is the first work that extends the well-known 
parity declustering technique (originally proposed for one-failure tolerant codes) for $\delta$-failure tolerant codes, 
for any $\delta \geq 1$. 

The paper is organized as follows. Necessary definitions and notations are provided in Section~\ref{sec:preliminaries}. 
In this section, we also review the parity declustering technique for one-failure tolerant codes based on BIBDs. 
We extend this technique for two-failure tolerant codes via $3$-designs in Section~\ref{sec:pd-2-failures}. 
In Section~\ref{sec:pd-(t-1)-failures}, we discuss the generalization of this idea for codes that can tolerate $\delta \geq 2$
disk failures. The paper is concluded in Section~\ref{sec:conclusion}.
  
\section{Preliminaries}
\label{sec:preliminaries}

Disk arrays spread data across several disks and access them in parallel to increase data transfer rates and I/O rates. 
Disk arrays are, however, highly vulnerable to disk failures. An array with $n$ disks is $n$ times more likely 
to fail than a single disk \cite{PattersonGibsonKatz1988}. Adding redundancy to a disk array is a natural solution
to this problem. \emph{Units} of data on $k$ disks are grouped together into \emph{parity groups} (or parity stripes).
Each parity group consists of $k-1$ data units and one parity unit. The parity unit is calculated by taking the XOR-sum
of the data units in the same group. The parity unit must be updated whenever a data unit in its group is modified. 
Therefore, the parity units should be distributed across the array rather than all being located on a small subset of disks. 
Otherwise we would have the situation where some disks are always busy updating the parity units while the others are 
totally idle. Ideally, we want to have the same number of parity units on every disk. This requirement guarantees that 
the parity update workload is uniformly distributed among all disks. Additionally, it is required that no two units from 
the same parity group are located on the same disk, so that the disk array can always be recovered from one disk failure. 

\begin{figure}[h]
\center
\scalebox{1} 
{
\begin{pspicture}(0,-2.41)(5.148,2.41)
\definecolor{color1553b}{rgb}{0.8,0.8,0.8}
\psframe[linewidth=0.04,dimen=outer](5.148,1.59)(0.0,-2.41)
\psline[linewidth=0.04cm](0.04,0.76)(5.128,0.79)
\psline[linewidth=0.04cm](0.04,-0.04)(5.128,-0.05)
\psline[linewidth=0.04cm](0.04,-0.84)(5.128,-0.85)
\psline[linewidth=0.04cm](1.04,1.5702401)(1.068,-2.39)
\psline[linewidth=0.04cm](2.04,1.5702401)(2.048,-2.39)
\psline[linewidth=0.04cm](3.04,1.5702401)(3.048,-2.37)
\usefont{T1}{ptm}{m}{n}
\rput(0.5079688,1.16){$D_0$}
\usefont{T1}{ptm}{m}{n}
\rput(1.5279686,1.16){$D_0$}
\usefont{T1}{ptm}{m}{n}
\rput(2.507969,1.16){$D_0$}
\usefont{T1}{ptm}{m}{n}
\rput(3.5479689,1.16){$D_0$}
\psframe[linewidth=0.04,dimen=outer,fillstyle=solid,fillcolor=color1553b](5.148,2.41)(0.008,1.69)
\psline[linewidth=0.04cm](1.028,1.71)(1.028,2.37)
\psline[linewidth=0.04cm](2.048,1.71)(2.048,2.37)
\psline[linewidth=0.04cm](3.048,1.71)(3.048,2.37)
\usefont{T1}{ptm}{m}{n}
\rput(0.5151562,2.02){Disk 0}
\usefont{T1}{ptm}{m}{n}
\rput(1.5082812,2.02){Disk 1}
\usefont{T1}{ptm}{m}{n}
\rput(2.531875,2.02){Disk 2}
\usefont{T1}{ptm}{m}{n}
\rput(3.549375,2.02){Disk 3}
\usefont{T1}{ptm}{m}{n}
\rput(0.5079688,0.34){$D_1$}
\usefont{T1}{ptm}{m}{n}
\rput(1.5279686,0.34){$D_1$}
\usefont{T1}{ptm}{m}{n}
\rput(2.5079687,0.34){$D_1$}
\usefont{T1}{ptm}{m}{n}
\rput(3.5279691,0.34){$P_1$}
\usefont{T1}{ptm}{m}{n}
\rput(0.5079688,-0.46){$D_2$}
\usefont{T1}{ptm}{m}{n}
\rput(1.5279686,-0.46){$D_2$}
\usefont{T1}{ptm}{m}{n}
\rput(2.487969,-0.46){$P_2$}
\usefont{T1}{ptm}{m}{n}
\rput(3.547969,-0.46){$D_2$}
\usefont{T1}{ptm}{m}{n}
\rput(0.5079688,-1.24){$D_3$}
\usefont{T1}{ptm}{m}{n}
\rput(1.5279686,-1.24){$P_3$}
\usefont{T1}{ptm}{m}{n}
\rput(2.507969,-1.24){$D_3$}
\usefont{T1}{ptm}{m}{n}
\rput(3.547969,-1.24){$D_3$}
\psline[linewidth=0.04cm](4.068,1.71)(4.068,2.37)
\usefont{T1}{ptm}{m}{n}
\rput(4.618125,2.02){Disk 4}
\psline[linewidth=0.04cm](4.08,1.5502402)(4.068,-2.39)
\usefont{T1}{ptm}{m}{n}
\rput(4.587969,1.16){$P_0$}
\usefont{T1}{ptm}{m}{n}
\rput(4.6079693,0.34){$D_1$}
\usefont{T1}{ptm}{m}{n}
\rput(4.6079693,-0.46){$D_2$}
\usefont{T1}{ptm}{m}{n}
\rput(4.6079693,-1.24){$D_3$}
\psline[linewidth=0.04cm](0.04,-1.62)(5.128,-1.63)
\usefont{T1}{ptm}{m}{n}
\rput(0.4879688,-2.04){$P_4$}
\usefont{T1}{ptm}{m}{n}
\rput(1.5479686,-2.04){$D_4$}
\usefont{T1}{ptm}{m}{n}
\rput(2.507969,-2.04){$D_4$}
\usefont{T1}{ptm}{m}{n}
\rput(3.547969,-2.04){$D_4$}
\usefont{T1}{ptm}{m}{n}
\rput(4.6079693,-2.04){$D_4$}
\end{pspicture} 
}
\caption{An array code with no parity declustering}
\label{fig:RAID-4}
\end{figure}   

Let us consider the following example. 
Suppose there are five disks in the disk array. Each disk is divided into several units. They are
either data units ($D$) or parity units ($P$). Each parity group consists of four data units
and one parity unit (those that have the same index). 
The parity unit is equal to the XOR-sum of the data units in 
the same parity group. 
The array in Fig.~\ref{fig:RAID-4} represents the \emph{basic} 
data/parity layout in this disk array.   
The basic layout is then repeated many times until every unit in each disk
is covered.
This data/parity layout is called an \emph{array code} for the disk array. 
Column~$i$ of the array code corresponds to Disk~$i$ in the disk array
that employs the array code. A data/parity entry in Column $i$ represents a data/parity unit in Disk~$i$.  
Without loss of generality, we assume that the disk array consists of only one
copy of the data/parity layout from the array code. 
In other words, we assume that the data/parity layout of the disk array looks completely 
the same as the data/parity layout of the array code. 
Then, throughout this work, we often use disks and columns, 
units and entries, interchangeably.  

The array code presented in Fig.~\ref{fig:RAID-4} can recover \emph{one} missing column. Hence, the disk array that employs this array code 
can tolerate \emph{one} disk failure. 
The reconstruction process of the lost column (disk) requires access to \emph{all} entries (units) in every surviving column (disk). 

The parity declustering technique for one-failure tolerant array codes based on BIBDs was originally suggested by Muntz and Lui \cite{MuntzLui1990} and investigated in details by Holland and Gibson \cite{HollandGibson1992}, Ng and Mattson \cite{NgMattson1992}, and Reddy and Banerjee \cite{ReddyBannerjee1991}. Before describing this technique, we need the definitions of $t$-designs and BIBDs. 

\vskip 10pt 
\begin{definition}
A $t$-$(n,k,\lam)$ \emph{design}, a $t$-design in short, is a pair $(\X, \B)$ where $\X$ is a set of $n$ \emph{points}
and $\B$ is a collection of $k$-subsets of $\X$ (\emph{blocks}) with the property that every $t$-subset of $\X$ is contained in exactly $\lam$ blocks. A $2$-$(n,k,\lam)$ design is also called a \emph{balanced incomplete block design} (BIBD). 
\end{definition} 

\begin{figure}[h]
\center
\scalebox{1} 
{
\begin{pspicture}(0,-2.065)(5.099841,2.045)
\definecolor{color1553b}{rgb}{0.8,0.8,0.8}
\psframe[linewidth=0.04,dimen=outer](5.08,1.225)(0.012,-2.045)
\psline[linewidth=0.04cm](0.056989312,0.395)(5.0798407,0.375)
\psline[linewidth=0.04cm](0.032,-0.405)(5.0798407,-0.405)
\psline[linewidth=0.04cm](0.056989312,-1.205)(5.0798407,-1.205)
\psline[linewidth=0.04cm](1.032,1.2052401)(1.052,-2.025)
\psline[linewidth=0.04cm](2.032,1.2052401)(2.052,-2.025)
\psline[linewidth=0.04cm](3.032,1.2052401)(3.052,-2.025)
\usefont{T1}{ptm}{m}{n}
\rput(0.51856256,0.795){$D_0$}
\usefont{T1}{ptm}{m}{n}
\rput(1.5385624,0.795){$D_0$}
\usefont{T1}{ptm}{m}{n}
\rput(2.5185626,0.795){$D_0$}
\usefont{T1}{ptm}{m}{n}
\rput(3.5385625,0.795){$P_0$}
\psframe[linewidth=0.04,dimen=outer,fillstyle=solid,fillcolor=color1553b](5.08,2.045)(0.0,1.325)
\psline[linewidth=0.04cm](1.02,1.345)(1.02,2.005)
\psline[linewidth=0.04cm](2.04,1.345)(2.04,2.005)
\psline[linewidth=0.04cm](3.04,1.345)(3.04,2.005)
\usefont{T1}{ptm}{m}{n}
\rput(0.51668745,1.655){Disk 0}
\usefont{T1}{ptm}{m}{n}
\rput(1.5254375,1.655){Disk 1}
\usefont{T1}{ptm}{m}{n}
\rput(2.5345,1.655){Disk 2}
\usefont{T1}{ptm}{m}{n}
\rput(3.5595,1.655){Disk 3}
\usefont{T1}{ptm}{m}{n}
\rput(0.51856256,-0.025){$D_1$}
\usefont{T1}{ptm}{m}{n}
\rput(1.5385624,-0.025){$D_1$}
\usefont{T1}{ptm}{m}{n}
\rput(2.5185626,-0.025){$D_1$}
\usefont{T1}{ptm}{m}{n}
\rput(3.5585628,-0.025){$D_2$}
\usefont{T1}{ptm}{m}{n}
\rput(0.51856256,-0.825){$D_2$}
\usefont{T1}{ptm}{m}{n}
\rput(1.5385624,-0.825){$D_2$}
\usefont{T1}{ptm}{m}{n}
\rput(2.5185626,-0.825){$D_3$}
\usefont{T1}{ptm}{m}{n}
\rput(3.5585628,-0.825){$D_3$}
\usefont{T1}{ptm}{m}{n}
\rput(0.51856256,-1.605){$D_3$}
\usefont{T1}{ptm}{m}{n}
\rput(1.5385624,-1.605){$D_4$}
\usefont{T1}{ptm}{m}{n}
\rput(2.5185626,-1.605){$D_4$}
\usefont{T1}{ptm}{m}{n}
\rput(3.5585628,-1.605){$D_4$}
\psline[linewidth=0.04cm](4.072,1.1852401)(4.092,-2.045)
\psline[linewidth=0.04cm](4.04,1.345)(4.04,2.005)
\usefont{T1}{ptm}{m}{n}
\rput(4.56825,1.655){Disk 4}
\usefont{T1}{ptm}{m}{n}
\rput(4.5585628,0.795){$P_1$}
\usefont{T1}{ptm}{m}{n}
\rput(4.5585628,-0.025){$P_2$}
\usefont{T1}{ptm}{m}{n}
\rput(4.5585628,-0.825){$P_3$}
\usefont{T1}{ptm}{m}{n}
\rput(4.5585628,-1.605){$P_4$}
\end{pspicture} 
}
\caption{An array code with parity declustering}
\label{fig:RAID-4-parity-declustering}
\end{figure}

Given a $2$-$(n,k,\lam)$ design, we associate disks with points and parity groups with blocks. 
As an illustrative example, consider a $2$-$(5,4,3)$ design with $\X = \{0,1,2,3,4\}$ and 
$\B$ consisting of five blocks: $\{0,1,2,3\}$,
$\{0,1,2,4\}$, $\{0,1,3,4\}$, $\{0,2,3,4\}$, $\{1,2,3,4\}$. 
Each block corresponds to one parity group. For instance, the block
$\{1,2,3,4\}$ corresponds to a parity group with the (three) data units being
located in Disks 1, 2, 3 and the parity unit located in Disk 4. The data layout of the array code
is presented in Fig.~\ref{fig:RAID-4-parity-declustering}. 
Furthermore, we can balance the number of parity units in every column by rotating
the array in this figure cyclically five times (see \cite{HollandGibson1992}).   

Since every two elements in the
set $\{0,1,2,3,4\}$ appears in precisely $\lam=3$ different blocks, every two disks share three
pairs of units, where units in each pair belong to the same parity group. Therefore, when 
one disk fails, precisely three units in each surviving disk need to be read 
for the recovery of units on the failed disk. 
Thus, instead of reading $100\%$ units in each surviving disk 
(as for the array code in Fig.~\ref{fig:RAID-4}), the reconstruction
process now reads $75\%$ units in each disk. In other words, by increasing the overhead
for the storage of parity (from a $1/5$ fraction of the space to a $1/4$ fraction), 
we can reduce the percentage of data that needs to be read in each surviving disk for recovery.  
However, we lose the MDS property of the code while spreading out the workload over more disks.
Now it requires $1.25$ disks worth of parity (see Section~\ref{subsec:trade-off} for a formal definition) 
instead of just one parity disk as in the previous example.   
Therefore, the parity declustering technique can be considered as a way to sacrifice the efficiency 
for faster reconstruction time. 

The connection between the reconstruction of one-disk failure and a $2$-design is elaborated further 
as follows. 
If a parity group $G$ contains a unit from a disk then that disk is said to be \emph{crossed} by $G$. 
The reconstruction of one unit requires access to all other units in the same parity group.  
Therefore, in order to have uniform workloads during the reconstruction for one disk failure, 
every two disks must share the same number of pairs of units that are from the same parity groups.
In other words, every two disks must be simultaneously crossed by the same number of parity groups.
If disks and parity groups are associated to points and blocks, respectively, then the aforementioned
property of the data layout becomes the familiar requirement for a $2$-design: every two points must
be simultaneously contained in the same number of blocks. Thus, the parity declustering technique for one-failure tolerant
array codes can be summarized as follows: \\

\nin {\bf Algorithm 1} (\cite{MuntzLui1990, HollandGibson1992, NgMattson1992, ReddyBannerjee1991})
\begin{itemize}
  \item {\bf Input:} $n$ is the number of physical disks in the array and $k$ is the parity group size.
	\item {\bf Step 1:} Choose a parity group $G$ with $k-1$ data units and one parity unit.
	\item {\bf Step 2:} Choose a $2$-$(n,k,\lam)$ design $\D = (\X,\B)$ for some $\lam$.
	\item {\bf Step 3:} For each block $B_i = \{b_{i,0},\ldots,b_{i,k-1}\} \in \B$, $0 \leq i < |\B|$, 
	create a parity group $G_i$ as follows. Firstly, $G_i$ must have the same data-parity pattern as $G$. 
	In other words, $G_i$
	has $k-1$ data units and one parity unit, and the parity unit is equal to the XOR-sum of the data units. 
	Secondly, the $k-1$ data units of $G_i$ are located on disks with labels $b_{i,0}, \ldots, b_{i,k-2}$. 
	The parity unit of $G_i$ is located on disk with label $b_{i,k-1}$.
	\item {\bf Output:} The $n$-disk array with $|\B|$ parity groups and their layouts according to Step~3.  
\end{itemize}

After employing Algorithm~1, as shown in \cite{HollandGibson1992}, the number of parity units in every column can be made balanced by rotating the resulting array
cyclically $n$ times.   

In the next sections, we generalize this procedure to construct declustered-parity layouts for 
array codes that tolerate more than one disk failure. 

\section{Parity Declustering for Two-Failure Tolerant Codes via $3$-Designs} 
\label{sec:pd-2-failures}

To extend the parity declustering technique for two-failure tolerant codes, 
we use balanced $2$-parity groups instead of parity groups. 

\subsection{$\delta$-Parity Groups}
\label{subsec:delta-parity-group}

\begin{definition} 
\label{def:pg}
A \emph{$\delta$-parity group} is an MDS $\delta$-failure tolerant array code. 
More formally, a \emph{$\delta$-parity group} is an $m \times k$ array that satisfies the following conditions:
\begin{itemize}
	\item[(C1)] it contains $(k - \delta)m$ data entries and $\delta m$ parity entries; 
  \item[(C2)] entries in at most $\delta$ columns can always be reconstructed from 
	the entries in other columns. 
\end{itemize}
Moreover, if a $\delta$-parity group also satisfies the two other conditions
\begin{itemize}
	\item[(C3)] for the reconstruction of entries in at most $\delta$ columns, the number of entries
	in every other column that contribute to the calculation must always be the same;
	\item[(C4)] the number of parity entries in every column must be the same, 
\end{itemize}
then it is said to be \emph{balanced}. If a $\delta$-parity group does not 
satisfy either (C3) or (C4) then it is said to be \emph{unbalanced}.
We refer to $k$ as the \emph{size} and $m$ as the \emph{depth}, 
respectively, of the $\delta$-parity group. 
\end{definition}
\vskip 10pt 

Note that the condition (C3) depends on the particular reconstruction algorithm 
used for the $\delta$-parity group. Therefore, a $\delta$-parity group can be 
balanced or unbalanced when different reconstruction algorithms are employed.    
In fact, all MDS two-failure tolerant array codes, such as Reed-Solomon (RS) codes \cite{Plank1997},
EVENODD \cite{BlaumBradyBruckMenon1995},
RDP \cite{Corbett2004}, B-code \cite{XuBohossianBruckWagner1998}, 
P-codes \cite{JinJiangFengTian2009}, X-codes \cite{XuBruck2008}, are $2$-parity groups. 
However, they are not yet balanced in their original form. 
The \emph{vertical} codes (B-, P-, X-codes), which contain \emph{both} data and parity units
in each column, 
equipped with their conventional reconstruction algorithms for
one failure, satisfy (C4) but not (C3). The \emph{horizontal} codes (RS, EVENODD, RDP), 
which contain \emph{either} data \emph{or} parity units in each column, in their original
form satisfy neither (C3) nor (C4). The following example shows how to modify 
the existing MDS horizontal codes to obtain balanced $2$-parity groups. 

\begin{example}
\label{ex:bpg-1}

We first consider RDP codes.
Let $p$ be a prime. RDP code for a $(p+1)$-disk array is defined as a $(p-1)\times(p+1)$
array \cite{Corbett2004} (see Fig~\ref{fig:RDP}). Its first $p-1$ columns (disks) store data entries (units) and its last two columns (disks)
store parity entries (units). The first parity column ($P$-column) stores the row-parity entries; each of such entries 
is equal to the XOR-sum of the data entries on the same row. The second parity column ($Q$-column)
stores the diagonal-parity entries; each of such entries is equal to the XOR-sum of the data 
and row-parity entries along some diagonal of the array. Note that one diagonal is not used (called 
the \emph{missing} diagonal in \cite{Corbett2004}). 

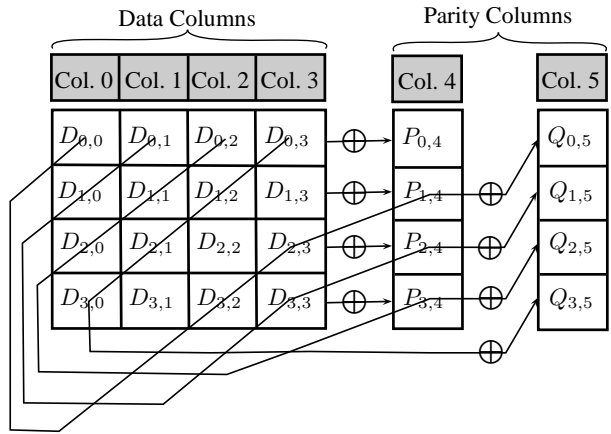
\begin{figure}[h]
	\scalebox{0.9} 
{
\begin{pspicture}(0,-3.137)(9.139844,3.165)
\definecolor{color1553b}{rgb}{0.8,0.8,0.8}
\psframe[linewidth=0.04,dimen=outer](4.68,1.645)(0.6,-1.625)
\psline[linewidth=0.04cm](0.64,0.815)(4.66,0.795)
\psline[linewidth=0.04cm](0.62,0.015)(4.66,0.015)
\psline[linewidth=0.04cm](0.64,-0.785)(4.66,-0.785)
\psline[linewidth=0.04cm](1.62,1.6252402)(1.64,-1.605)
\psline[linewidth=0.04cm](2.62,1.6252402)(2.64,-1.605)
\psline[linewidth=0.04cm](3.62,1.6252402)(3.64,-1.605)
\psframe[linewidth=0.04,dimen=outer](6.7,1.645)(5.64,-1.6286318)
\psframe[linewidth=0.04,dimen=outer](8.84,1.645)(7.78,-1.645)
\psline[linewidth=0.04cm](5.68,0.775)(6.68,0.775)
\psline[linewidth=0.04cm](5.68,-0.005)(6.68,-0.005)
\psline[linewidth=0.04cm](5.68,-0.805)(6.68,-0.805)
\psline[linewidth=0.04cm](7.82,0.795)(8.82,0.795)
\psline[linewidth=0.04cm](7.82,-0.005)(8.82,-0.005)
\psline[linewidth=0.04cm](7.82,-0.805)(8.82,-0.805)
\usefont{T1}{ptm}{m}{n}
\rput(1.0493749,1.215){$D_{0,0}$}
\usefont{T1}{ptm}{m}{n}
\rput(1.0293748,0.395){$D_{1,0}$}
\usefont{T1}{ptm}{m}{n}
\rput(1.0693749,-0.365){$D_{2,0}$}
\usefont{T1}{ptm}{m}{n}
\rput(1.0693749,-1.185){$D_{3,0}$}
\usefont{T1}{ptm}{m}{n}
\rput(2.0693748,1.215){$D_{0,1}$}
\usefont{T1}{ptm}{m}{n}
\rput(2.0493748,0.395){$D_{1,1}$}
\usefont{T1}{ptm}{m}{n}
\rput(2.0893748,-0.365){$D_{2,1}$}
\usefont{T1}{ptm}{m}{n}
\rput(2.0893748,-1.185){$D_{3,1}$}
\usefont{T1}{ptm}{m}{n}
\rput(3.0493748,1.215){$D_{0,2}$}
\usefont{T1}{ptm}{m}{n}
\rput(3.0493748,0.395){$D_{1,2}$}
\usefont{T1}{ptm}{m}{n}
\rput(3.0893748,-0.365){$D_{2,2}$}
\usefont{T1}{ptm}{m}{n}
\rput(3.0893748,-1.185){$D_{3,2}$}
\usefont{T1}{ptm}{m}{n}
\rput(4.089375,1.215){$D_{0,3}$}
\usefont{T1}{ptm}{m}{n}
\rput(4.089375,0.395){$D_{1,3}$}
\usefont{T1}{ptm}{m}{n}
\rput(4.1293755,-0.365){$D_{2,3}$}
\usefont{T1}{ptm}{m}{n}
\rput(4.1293755,-1.185){$D_{3,3}$}
\usefont{T1}{ptm}{m}{n}
\rput(6.089375,1.195){$P_{0,4}$}
\usefont{T1}{ptm}{m}{n}
\rput(6.1093755,0.395){$P_{1,4}$}
\usefont{T1}{ptm}{m}{n}
\rput(6.1093755,-0.365){$P_{2,4}$}
\usefont{T1}{ptm}{m}{n}
\rput(6.1093755,-1.185){$P_{3,4}$}
\usefont{T1}{ptm}{m}{n}
\rput(8.269375,1.215){$Q_{0,5}$}
\usefont{T1}{ptm}{m}{n}
\rput(8.289375,0.395){$Q_{1,5}$}
\usefont{T1}{ptm}{m}{n}
\rput(8.289375,-0.365){$Q_{2,5}$}
\usefont{T1}{ptm}{m}{n}
\rput(8.289375,-1.185){$Q_{3,5}$}
\psline[linewidth=0.024,arrowsize=0.05291667cm 2.0,arrowlength=1.4,arrowinset=0.4]{->}(1.1,1.215)(0.0,0.275)(0.0,-3.105)(0.72,-3.125)(4.14,-0.405)(6.2,0.375)(7.32,0.375)(7.8,1.195)
\usefont{T1}{ptm}{m}{n}
\rput(7.1056247,0.385){$\bigoplus$}
\psline[linewidth=0.024,arrowsize=0.05291667cm 2.0,arrowlength=1.4,arrowinset=0.4]{->}(2.14,1.215)(0.18,-0.345)(0.2,-2.705)(2.26,-2.725)(4.12,-1.185)(6.18,-0.405)(7.3,-0.405)(7.78,0.415)
\usefont{T1}{ptm}{m}{n}
\rput(7.0856247,-0.415){$\bigoplus$}
\psline[linewidth=0.024,arrowsize=0.05291667cm 2.0,arrowlength=1.4,arrowinset=0.4]{->}(3.18,1.195)(0.4,-0.965)(0.42,-2.245)(3.2,-2.285)(4.14,-1.945)(6.2,-1.165)(7.32,-1.165)(7.8,-0.345)
\usefont{T1}{ptm}{m}{n}
\rput(7.1056247,-1.155){$\bigoplus$}
\psline[linewidth=0.024,arrowsize=0.05291667cm 2.0,arrowlength=1.4,arrowinset=0.4]{->}(4.12,1.215)(1.16,-1.205)(1.18,-1.945)(2.96,-1.965)(4.42,-1.945)(6.2,-1.965)(7.32,-1.965)(7.8,-1.145)
\usefont{T1}{ptm}{m}{n}
\rput(7.1056247,-1.955){$\bigoplus$}
\usefont{T1}{ptm}{m}{n}
\rput(5.0856247,1.165){$\bigoplus$}
\psline[linewidth=0.024cm,arrowsize=0.05291667cm 2.0,arrowlength=1.4,arrowinset=0.4]{->}(4.66,1.155)(5.64,1.175)
\usefont{T1}{ptm}{m}{n}
\rput(5.0856247,0.405){$\bigoplus$}
\psline[linewidth=0.024cm,arrowsize=0.05291667cm 2.0,arrowlength=1.4,arrowinset=0.4]{->}(4.66,0.395)(5.64,0.415)
\usefont{T1}{ptm}{m}{n}
\rput(5.0856247,-0.395){$\bigoplus$}
\psline[linewidth=0.024cm,arrowsize=0.05291667cm 2.0,arrowlength=1.4,arrowinset=0.4]{->}(4.66,-0.405)(5.64,-0.385)
\usefont{T1}{ptm}{m}{n}
\rput(5.0856247,-1.215){$\bigoplus$}
\psline[linewidth=0.024cm,arrowsize=0.05291667cm 2.0,arrowlength=1.4,arrowinset=0.4]{->}(4.66,-1.225)(5.64,-1.205)
\psbezier[linewidth=0.02](2.6275833,2.8111668)(2.5564165,2.5072575)(0.44141656,2.817332)(0.62861794,2.5336673)
\psbezier[linewidth=0.02](2.6103694,2.8112152)(2.7112622,2.4963343)(4.8998194,2.8047588)(4.656753,2.532794)
\usefont{T1}{ptm}{m}{n}
\rput(2.6125,2.985){Data Columns}
\psbezier[linewidth=0.02](7.2172103,2.7711668)(7.162007,2.4672575)(5.5214167,2.777332)(5.6666274,2.4936674)
\psbezier[linewidth=0.02](7.2038574,2.7712152)(7.2821193,2.4563344)(8.979768,2.7647586)(8.791223,2.4927938)
\usefont{T1}{ptm}{m}{n}
\rput(7.203125,2.985){Parity Columns}
\psframe[linewidth=0.04,dimen=outer,fillstyle=solid,fillcolor=color1553b](4.668,2.465)(0.588,1.745)
\psline[linewidth=0.04cm](1.608,1.765)(1.608,2.425)
\psline[linewidth=0.04cm](2.628,1.765)(2.628,2.425)
\psline[linewidth=0.04cm](3.628,1.765)(3.628,2.425)
\usefont{T1}{ptm}{m}{n}
\rput(1.08875,2.075){Col. 0}
\usefont{T1}{ptm}{m}{n}
\rput(2.1262498,2.075){Col. 1}
\usefont{T1}{ptm}{m}{n}
\rput(3.104375,2.075){Col. 2}
\usefont{T1}{ptm}{m}{n}
\rput(4.134375,2.075){Col. 3}
\psframe[linewidth=0.04,dimen=outer,fillstyle=solid,fillcolor=color1553b](6.688,2.445)(5.628,1.725)
\usefont{T1}{ptm}{m}{n}
\rput(6.160625,2.055){Col. 4}
\psframe[linewidth=0.04,dimen=outer,fillstyle=solid,fillcolor=color1553b](8.828,2.465)(7.768,1.745)
\usefont{T1}{ptm}{m}{n}
\rput(8.28375,2.075){Col. 5}
\end{pspicture} 
}
\caption{RDP array with $p = 5$ (reproduced from \cite{Xiang-et-al2011})}
	\label{fig:RDP}
\end{figure}

Below we show that the RDP array is not a balanced $2$-parity group.  
The reconstruction rule for RDP (\cite{Corbett2004}) is as follows. 
Suppose one column is lost. 
If it is a data column ($D$), then each of its entries can be recovered by taking the XOR-sum of 
the data entries in other data columns ($D$) and the row parity entry on the $P$-column 
that belong to the same row. In this way, the $Q$-column plays no role in the reconstruction of 
one lost data column. If the $P$-column or the $Q$-column is lost, then its entries can be reconstructed 
by recalculating the parities according to the encoding rule of RDP. Note that the reconstruction
of the $P$-column does \emph{not} require access to the $Q$-column, and vice versa. 
Hence, the RDP array and its conventional reconstruction rule does not qualify
as a balanced $2$-parity group. 

However, we can transform an RDP array into a
balanced $2$-parity group as follows.
Let us first label the data columns by '$D$' and the parity columns by '$P$' and '$Q$', respectively. 
As an example, the RDP array ($p = 5$) in its simplified layout is depicted in Fig.~\ref{fig:RDP-simple}. 

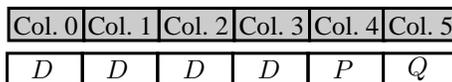
\begin{figure}[h]
\centering
\scalebox{1}
{
\begin{pspicture}(0,-0.535)(6.04,0.535)
\definecolor{color1553b}{rgb}{0.8,0.8,0.8}
\usefont{T1}{ptm}{m}{n}
\rput(0.48221874,-0.305){$D$}
\usefont{T1}{ptm}{m}{n}
\rput(1.4822187,-0.305){$D$}
\usefont{T1}{ptm}{m}{n}
\rput(2.5022187,-0.305){$D$}
\usefont{T1}{ptm}{m}{n}
\rput(3.5022187,-0.305){$D$}
\usefont{T1}{ptm}{m}{n}
\rput(4.4822187,-0.305){$P$}
\usefont{T1}{ptm}{m}{n}
\rput(5.4822187,-0.285){$Q$}
\psframe[linewidth=0.04,dimen=outer,fillstyle=solid,fillcolor=color1553b](1.06,0.535)(0.02,0.065)
\psframe[linewidth=0.04,dimen=outer,fillstyle=solid,fillcolor=color1553b](2.06,0.535)(1.02,0.065)
\psframe[linewidth=0.04,dimen=outer,fillstyle=solid,fillcolor=color1553b](3.06,0.535)(2.02,0.065)
\psframe[linewidth=0.04,dimen=outer,fillstyle=solid,fillcolor=color1553b](4.06,0.535)(3.02,0.065)
\psframe[linewidth=0.04,dimen=outer,fillstyle=solid,fillcolor=color1553b](5.06,0.535)(4.02,0.065)
\psframe[linewidth=0.04,dimen=outer](1.04,-0.065)(0.0,-0.535)
\psframe[linewidth=0.04,dimen=outer](2.04,-0.065)(1.0,-0.535)
\psframe[linewidth=0.04,dimen=outer](3.04,-0.065)(2.0,-0.535)
\psframe[linewidth=0.04,dimen=outer](4.04,-0.065)(3.0,-0.535)
\psframe[linewidth=0.04,dimen=outer](5.04,-0.065)(4.0,-0.535)
\psframe[linewidth=0.04,dimen=outer,fillstyle=solid,fillcolor=color1553b](6.04,0.535)(5.0,0.065)
\psframe[linewidth=0.04,dimen=outer](6.04,-0.065)(5.0,-0.535)
\usefont{T1}{ptm}{m}{n}
\rput(0.53406245,0.285){Col. 0}
\usefont{T1}{ptm}{m}{n}
\rput(1.5228126,0.285){Col. 1}
\usefont{T1}{ptm}{m}{n}
\rput(2.531875,0.285){Col. 2}
\usefont{T1}{ptm}{m}{n}
\rput(3.536875,0.285){Col. 3}
\usefont{T1}{ptm}{m}{n}
\rput(4.5343747,0.285){Col. 4}
\usefont{T1}{ptm}{m}{n}
\rput(5.523125,0.285){Col. 5}
\end{pspicture} 
}
\caption{An RDP array with $p=5$ (simplified layout)}
\label{fig:RDP-simple}
\end{figure}

We consider all possible ways to arrange the $P$-column and the $Q$-column
among all $k$ columns ($k = p+1$). There are $k(k-1)$
such arrangements.
If $k = 6$ then there are $30 = 6 \times 5$ possible such arrangements.  
For each of such arrangements of $P$- and $Q$-columns, we obtain a new array, $\A_i$, $0 \leq i < k(k-1)$. 
We juxtapose all these arrays vertically to obtain a new array $\G$, which contains 
$k(k-1)$ times more rows than the original RDP array (see Fig.~\ref{fig:rdp-juxtapose} for the case when $k = 6$). 

\begin{figure}[h]
\centering
\scalebox{1} 
{
\begin{pspicture}(0,-2.415)(8.157187,2.415)
\definecolor{color1984b}{rgb}{0.8,0.8,0.8}
\usefont{T1}{ptm}{m}{n}
\rput(2.5794063,1.575){$D$}
\usefont{T1}{ptm}{m}{n}
\rput(3.5794063,1.575){$D$}
\usefont{T1}{ptm}{m}{n}
\rput(4.5994062,1.575){$D$}
\usefont{T1}{ptm}{m}{n}
\rput(5.5994062,1.575){$D$}
\usefont{T1}{ptm}{m}{n}
\rput(6.5794063,1.575){$P$}
\usefont{T1}{ptm}{m}{n}
\rput(7.5794067,1.575){$Q$}
\psframe[linewidth=0.04,dimen=outer,fillstyle=solid,fillcolor=color1984b](3.1757812,2.415)(2.1357813,1.945)
\psframe[linewidth=0.04,dimen=outer,fillstyle=solid,fillcolor=color1984b](4.1757812,2.415)(3.1357813,1.945)
\psframe[linewidth=0.04,dimen=outer,fillstyle=solid,fillcolor=color1984b](5.1757812,2.415)(4.1357813,1.945)
\psframe[linewidth=0.04,dimen=outer,fillstyle=solid,fillcolor=color1984b](6.1757812,2.415)(5.1357813,1.945)
\psframe[linewidth=0.04,dimen=outer,fillstyle=solid,fillcolor=color1984b](7.1757812,2.415)(6.1357813,1.945)
\psframe[linewidth=0.04,dimen=outer](3.1557813,1.815)(2.1157813,1.345)
\psframe[linewidth=0.04,dimen=outer](4.1557813,1.815)(3.1157813,1.345)
\psframe[linewidth=0.04,dimen=outer](5.1557813,1.815)(4.1157813,1.345)
\psframe[linewidth=0.04,dimen=outer](6.1557813,1.815)(5.1157813,1.345)
\psframe[linewidth=0.04,dimen=outer](7.1557813,1.815)(6.1157813,1.345)
\psframe[linewidth=0.04,dimen=outer,fillstyle=solid,fillcolor=color1984b](8.155782,2.415)(7.1157813,1.945)
\psframe[linewidth=0.04,dimen=outer](8.155782,1.815)(7.1157813,1.345)
\usefont{T1}{ptm}{m}{n}
\rput(2.616875,2.165){Col. 0}
\usefont{T1}{ptm}{m}{n}
\rput(3.61,2.165){Col. 1}
\usefont{T1}{ptm}{m}{n}
\rput(4.613594,2.165){Col. 2}
\usefont{T1}{ptm}{m}{n}
\rput(5.6110935,2.165){Col. 3}
\usefont{T1}{ptm}{m}{n}
\rput(6.6173434,2.165){Col. 4}
\usefont{T1}{ptm}{m}{n}
\rput(7.600469,2.165){Col. 5}
\usefont{T1}{ptm}{m}{n}
\rput(0.90234375,1.55){$DDDDPQ$}
\usefont{T1}{ptm}{m}{n}
\rput(0.90234375,1.11){$DDDDQP$}
\usefont{T1}{ptm}{m}{n}
\rput(0.90234375,0.67){$DDDPDQ$}
\usefont{T1}{ptm}{m}{n}
\rput(0.90234375,0.23){$DDDQDP$}
\usefont{T1}{ptm}{m}{n}
\rput(0.90234375,-0.21){$DDPDDQ$}
\usefont{T1}{ptm}{m}{n}
\rput(0.90234375,-0.65){$DDQDDP$}
\usefont{T1}{ptm}{m}{n}
\rput(0.90234375,-1.71){$PQDDDD$}
\usefont{T1}{ptm}{m}{n}
\rput(0.90234375,-2.15){$QPDDDD$}
\usefont{T1}{ptm}{m}{n}
\rput(2.5794063,1.135){$D$}
\usefont{T1}{ptm}{m}{n}
\rput(3.5794063,1.135){$D$}
\usefont{T1}{ptm}{m}{n}
\rput(4.5994062,1.135){$D$}
\usefont{T1}{ptm}{m}{n}
\rput(5.5994062,1.135){$D$}
\usefont{T1}{ptm}{m}{n}
\rput(6.5994062,1.135){$Q$}
\usefont{T1}{ptm}{m}{n}
\rput(7.5594063,1.155){$P$}
\psframe[linewidth=0.04,dimen=outer](3.1557813,1.375)(2.1157813,0.905)
\psframe[linewidth=0.04,dimen=outer](4.1557813,1.375)(3.1157813,0.905)
\psframe[linewidth=0.04,dimen=outer](5.1557813,1.375)(4.1157813,0.905)
\psframe[linewidth=0.04,dimen=outer](6.1557813,1.375)(5.1157813,0.905)
\psframe[linewidth=0.04,dimen=outer](7.1557813,1.375)(6.1157813,0.905)
\psframe[linewidth=0.04,dimen=outer](8.155782,1.375)(7.1157813,0.905)
\usefont{T1}{ptm}{m}{n}
\rput(2.5794063,0.695){$D$}
\usefont{T1}{ptm}{m}{n}
\rput(3.5794063,0.695){$D$}
\usefont{T1}{ptm}{m}{n}
\rput(4.5994062,0.695){$D$}
\usefont{T1}{ptm}{m}{n}
\rput(6.5794063,0.675){$D$}
\usefont{T1}{ptm}{m}{n}
\rput(5.5994062,0.695){$P$}
\usefont{T1}{ptm}{m}{n}
\rput(7.5794067,0.695){$Q$}
\psframe[linewidth=0.04,dimen=outer](3.1557813,0.935)(2.1157813,0.465)
\psframe[linewidth=0.04,dimen=outer](4.1557813,0.935)(3.1157813,0.465)
\psframe[linewidth=0.04,dimen=outer](5.1557813,0.935)(4.1157813,0.465)
\psframe[linewidth=0.04,dimen=outer](6.1557813,0.935)(5.1157813,0.465)
\psframe[linewidth=0.04,dimen=outer](7.1557813,0.935)(6.1157813,0.465)
\psframe[linewidth=0.04,dimen=outer](8.155782,0.935)(7.1157813,0.465)
\usefont{T1}{ptm}{m}{n}
\rput(2.5794063,0.255){$D$}
\usefont{T1}{ptm}{m}{n}
\rput(3.5794063,0.255){$D$}
\usefont{T1}{ptm}{m}{n}
\rput(4.5994062,0.255){$D$}
\usefont{T1}{ptm}{m}{n}
\rput(6.5794063,0.255){$D$}
\usefont{T1}{ptm}{m}{n}
\rput(7.5994062,0.235){$P$}
\usefont{T1}{ptm}{m}{n}
\rput(5.5794067,0.255){$Q$}
\psframe[linewidth=0.04,dimen=outer](3.1557813,0.495)(2.1157813,0.025)
\psframe[linewidth=0.04,dimen=outer](4.1557813,0.495)(3.1157813,0.025)
\psframe[linewidth=0.04,dimen=outer](5.1557813,0.495)(4.1157813,0.025)
\psframe[linewidth=0.04,dimen=outer](6.1557813,0.495)(5.1157813,0.025)
\psframe[linewidth=0.04,dimen=outer](7.1557813,0.495)(6.1157813,0.025)
\psframe[linewidth=0.04,dimen=outer](8.155782,0.495)(7.1157813,0.025)
\usefont{T1}{ptm}{m}{n}
\rput(2.5794063,-0.185){$D$}
\usefont{T1}{ptm}{m}{n}
\rput(3.5794063,-0.185){$D$}
\usefont{T1}{ptm}{m}{n}
\rput(6.5794063,-0.185){$D$}
\usefont{T1}{ptm}{m}{n}
\rput(5.5994062,-0.185){$D$}
\usefont{T1}{ptm}{m}{n}
\rput(4.5794063,-0.185){$P$}
\usefont{T1}{ptm}{m}{n}
\rput(7.5794067,-0.185){$Q$}
\psframe[linewidth=0.04,dimen=outer](3.1557813,0.055)(2.1157813,-0.415)
\psframe[linewidth=0.04,dimen=outer](4.1557813,0.055)(3.1157813,-0.415)
\psframe[linewidth=0.04,dimen=outer](5.1557813,0.055)(4.1157813,-0.415)
\psframe[linewidth=0.04,dimen=outer](6.1557813,0.055)(5.1157813,-0.415)
\psframe[linewidth=0.04,dimen=outer](7.1557813,0.055)(6.1157813,-0.415)
\psframe[linewidth=0.04,dimen=outer](8.155782,0.055)(7.1157813,-0.415)
\usefont{T1}{ptm}{m}{n}
\rput(2.5794063,-0.625){$D$}
\usefont{T1}{ptm}{m}{n}
\rput(3.5794063,-0.625){$D$}
\usefont{T1}{ptm}{m}{n}
\rput(6.5794063,-0.625){$D$}
\usefont{T1}{ptm}{m}{n}
\rput(5.5994062,-0.625){$D$}
\usefont{T1}{ptm}{m}{n}
\rput(7.5994062,-0.625){$P$}
\usefont{T1}{ptm}{m}{n}
\rput(4.5794067,-0.625){$Q$}
\psframe[linewidth=0.04,dimen=outer](3.1557813,-0.385)(2.1157813,-0.855)
\psframe[linewidth=0.04,dimen=outer](4.1557813,-0.385)(3.1157813,-0.855)
\psframe[linewidth=0.04,dimen=outer](5.1557813,-0.385)(4.1157813,-0.855)
\psframe[linewidth=0.04,dimen=outer](6.1557813,-0.385)(5.1157813,-0.855)
\psframe[linewidth=0.04,dimen=outer](7.1557813,-0.385)(6.1157813,-0.855)
\psframe[linewidth=0.04,dimen=outer](8.155782,-0.385)(7.1157813,-0.855)
\usefont{T1}{ptm}{m}{n}
\rput(6.5594063,-1.745){$D$}
\usefont{T1}{ptm}{m}{n}
\rput(7.5594063,-1.745){$D$}
\usefont{T1}{ptm}{m}{n}
\rput(4.5994062,-1.745){$D$}
\usefont{T1}{ptm}{m}{n}
\rput(5.5994062,-1.745){$D$}
\usefont{T1}{ptm}{m}{n}
\rput(2.5794063,-1.745){$P$}
\usefont{T1}{ptm}{m}{n}
\rput(3.5794065,-1.745){$Q$}
\psframe[linewidth=0.04,dimen=outer](3.1557813,-1.505)(2.1157813,-1.975)
\psframe[linewidth=0.04,dimen=outer](4.1557813,-1.505)(3.1157813,-1.975)
\psframe[linewidth=0.04,dimen=outer](5.1557813,-1.505)(4.1157813,-1.975)
\psframe[linewidth=0.04,dimen=outer](6.1557813,-1.505)(5.1157813,-1.975)
\psframe[linewidth=0.04,dimen=outer](7.1557813,-1.505)(6.1157813,-1.975)
\psframe[linewidth=0.04,dimen=outer](8.155782,-1.505)(7.1157813,-1.975)
\usefont{T1}{ptm}{m}{n}
\rput(6.5794063,-2.185){$D$}
\usefont{T1}{ptm}{m}{n}
\rput(7.5794063,-2.185){$D$}
\usefont{T1}{ptm}{m}{n}
\rput(4.5994062,-2.185){$D$}
\usefont{T1}{ptm}{m}{n}
\rput(5.5994062,-2.185){$D$}
\usefont{T1}{ptm}{m}{n}
\rput(3.5794063,-2.205){$P$}
\usefont{T1}{ptm}{m}{n}
\rput(2.5794065,-2.185){$Q$}
\psframe[linewidth=0.04,dimen=outer](3.1557813,-1.945)(2.1157813,-2.415)
\psframe[linewidth=0.04,dimen=outer](4.1557813,-1.945)(3.1157813,-2.415)
\psframe[linewidth=0.04,dimen=outer](5.1557813,-1.945)(4.1157813,-2.415)
\psframe[linewidth=0.04,dimen=outer](6.1557813,-1.945)(5.1157813,-2.415)
\psframe[linewidth=0.04,dimen=outer](7.1557813,-1.945)(6.1157813,-2.415)
\psframe[linewidth=0.04,dimen=outer](8.155782,-1.945)(7.1157813,-2.415)
\psframe[linewidth=0.04,dimen=outer](8.157187,-0.83)(2.1171875,-1.55)
\psline[linewidth=0.08cm,linestyle=dotted,dotsep=0.16cm](5.1,-0.925)(5.12,-1.485)
\psline[linewidth=0.08cm,linestyle=dotted,dotsep=0.16cm](0.92,-0.945)(0.94,-1.505)
\end{pspicture} 
}
\caption{The simplified layout of a balanced $2$-parity group $\G$ obtained from an RDP array ($p=5$)}
\label{fig:rdp-juxtapose}
\end{figure}

Our goal now is to show that the array $\G$ constructed above, together 
with RDP's reconstruction rule (\cite{Corbett2004}), in general, 
is a balanced $2$-parity group. 
The array $\G$ obviously satisfies (C1), (C2), and (C4). 
We only need to verify Condition (C3) for $\G$. 
To recover two missing columns in an RDP array, every other column has to 
be read in full. Hence, the reconstruction workload
for two missing column is already uniform across the columns
of $\G$. 
To recover one missing column in an RDP array, each of other columns either has to be
read in full or is not accessed at all. Therefore, it suffices to 
regard each (RDP) column $D$, $P$, or $Q$ as a single entry, or more
precisely, a \emph{column-entry}, in $\G$, and use the reconstruction rule 
for RDP as shown in Fig.~\ref{fig:rdp-reconstruction}. 
Those column-entries of $\G$ correspond to \emph{column-units}
on physical disks where each column-unit is a column
of data/parity units. 

We refer to each $\A_i$ $(1 \leq i \leq k(k-1))$ as an \emph{extended-row} of $\G$.  
Then $\G$ has $k(k-1)$ extended-rows and each extended row contains $k$ 
column-entries. 
For instance, in Fig.~\ref{fig:rdp-juxtapose}, $\G$ has $30$
extended-rows and each extended-row contains $6$ column-entries.  

For two distinct columns $i$ and $j$ of $\G$, we define the following quantities: 
\begin{itemize}
	\item $r_{DQ}$: the number of extended-rows that has a $D$ at Column $i$ and has a $Q$ at Column $j$; 
	\item $r_{PQ}$: the number of extended-rows that has a $P$ at Column $i$ and has a $Q$ at Column $j$; 
	\item $r_{QP}$: the number of extended-rows that has a $Q$ at Column $i$ and has a $P$ at Column $j$.	
\end{itemize}
\begin{figure}[h]
\centering
		\begin{tabular}{|l|l|l|}
		\hline
		Lost & To be accessed & Not to be accessed \\
		\hline
		$D$ & $D$, $P$ & $Q$ \\
		\hline
		$P$ & $D$ & $Q$ \\
		\hline
		$Q$ & $D$ & $P$ \\
		\hline
		\end{tabular}
		\caption{Reconstruction rule for an RDP array}
		\label{fig:rdp-reconstruction}
\end{figure}
According to the reconstruction rule of RDP arrays (Fig.~\ref{fig:rdp-reconstruction}), 
these extended-rows (that define $r_{DQ}$, $r_{PQ}$, and $r_{QP}$ as above) are \emph{precisely} 
the extended-rows of $\G$ on which the recovery of the column-entry in the $i$th column does not
require access to the column-entry in the $j$th column. 
Therefore, the number of column-entries to be read in column $j$ during the reconstruction of column $i$ is precisely
\[
k(k-1) - r_{DQ} - r_{PQ} - r_{QP}. 
\]
Hence, if $r_{DQ}$, $r_{PQ}$, and $r_{QP}$ are all 
constants for every pair $(i, j)$ then the reconstruction workload is uniformly distributed to all surviving columns. 
As the extended-rows of $\G$ correspond to all possible arrangements of $P$-, $Q$-, and $D$-columns in an RDP
array of size $k$, we have
\[
r_{DQ} = k - 2, \ r_{PQ} = 1, \ r_{QP} = 1, 
\]
for every pair of columns $i$ and $j$ of $\G$. Therefore, $\G$
satisfies (C3). 

The same modification also turns an EVENODD array code or 
an RS code into a balanced $2$-parity group. 
In fact, this method works for every horizontal array code, as long
as they have separate parity columns ($P$- and $Q$-columns) and have 
reconstruction rules that can be clearly stated in tables similar to the one in 
Fig.~\ref{fig:rdp-reconstruction}. 
\end{example}
\vskip 10pt 

Note that a simple
cyclic rotation does \emph{not} turn a horizontal array code into a balanced
$2$-parity group. For instance, consider an array obtained 
by juxtaposing vertically all cyclic rotations of an RDP array
with $p = 5$ as in Fig~\ref{fig:rdp-rotate}. Suppose the first 
column is lost. For reconstruction, according to the rule illustrated in Fig.~\ref{fig:rdp-reconstruction}, one needs to access \emph{five} column-entries on the second column and only \emph{four} column-entries on the last column. Hence, the reconstruction workload is not distributed uniformly among the surviving columns. 

\vskip 10pt 
\begin{figure}[h]
\centering
\scalebox{1}
{
\begin{pspicture}(0,-1.650625)(8.077188,1.635)
\definecolor{color1553b}{rgb}{0.8,0.8,0.8}
\usefont{T1}{ptm}{m}{n}
\rput(2.5194063,0.795){$D$}
\usefont{T1}{ptm}{m}{n}
\rput(3.5194063,0.795){$D$}
\usefont{T1}{ptm}{m}{n}
\rput(4.5394063,0.795){$D$}
\usefont{T1}{ptm}{m}{n}
\rput(5.5394063,0.795){$D$}
\usefont{T1}{ptm}{m}{n}
\rput(6.5194063,0.795){$P$}
\usefont{T1}{ptm}{m}{n}
\rput(7.5194063,0.795){$Q$}
\psframe[linewidth=0.04,dimen=outer,fillstyle=solid,fillcolor=color1553b](3.0971875,1.635)(2.0571876,1.165)
\psframe[linewidth=0.04,dimen=outer,fillstyle=solid,fillcolor=color1553b](4.0971875,1.635)(3.0571876,1.165)
\psframe[linewidth=0.04,dimen=outer,fillstyle=solid,fillcolor=color1553b](5.0971875,1.635)(4.0571876,1.165)
\psframe[linewidth=0.04,dimen=outer,fillstyle=solid,fillcolor=color1553b](6.0971875,1.635)(5.0571876,1.165)
\psframe[linewidth=0.04,dimen=outer,fillstyle=solid,fillcolor=color1553b](7.0971875,1.635)(6.0571876,1.165)
\psframe[linewidth=0.04,dimen=outer](3.0771875,1.035)(2.0371876,0.565)
\psframe[linewidth=0.04,dimen=outer](4.0771875,1.035)(3.0371876,0.565)
\psframe[linewidth=0.04,dimen=outer](5.0771875,1.035)(4.0371876,0.565)
\psframe[linewidth=0.04,dimen=outer](6.0771875,1.035)(5.0371876,0.565)
\psframe[linewidth=0.04,dimen=outer](7.0771875,1.035)(6.0371876,0.565)
\psframe[linewidth=0.04,dimen=outer,fillstyle=solid,fillcolor=color1553b](8.077188,1.635)(7.0371876,1.165)
\psframe[linewidth=0.04,dimen=outer](8.077188,1.035)(7.0371876,0.565)
\usefont{T1}{ptm}{m}{n}
\rput(2.55125,1.385){Col. 0}
\usefont{T1}{ptm}{m}{n}
\rput(3.56,1.385){Col. 1}
\usefont{T1}{ptm}{m}{n}
\rput(4.5490627,1.385){Col. 2}
\usefont{T1}{ptm}{m}{n}
\rput(5.5340624,1.385){Col. 3}
\usefont{T1}{ptm}{m}{n}
\rput(6.5515623,1.385){Col. 4}
\usefont{T1}{ptm}{m}{n}
\rput(7.540313,1.385){Col. 5}
\usefont{T1}{ptm}{m}{n}
\rput(0.90234375,0.77){$DDDDPQ$}
\usefont{T1}{ptm}{m}{n}
\rput(0.90234375,0.33){$QDDDDP$}
\usefont{T1}{ptm}{m}{n}
\rput(0.90234375,-0.11){$PQDDDD$}
\usefont{T1}{ptm}{m}{n}
\rput(0.90234375,-0.55){$DPQDDD$}
\usefont{T1}{ptm}{m}{n}
\rput(0.90234375,-0.99){$DDPQDD$}
\usefont{T1}{ptm}{m}{n}
\rput(0.90234375,-1.43){$DDDPQD$}
\usefont{T1}{ptm}{m}{n}
\rput(6.5194063,0.355){$D$}
\usefont{T1}{ptm}{m}{n}
\rput(3.5194063,0.355){$D$}
\usefont{T1}{ptm}{m}{n}
\rput(4.5394063,0.355){$D$}
\usefont{T1}{ptm}{m}{n}
\rput(5.5394063,0.355){$D$}
\usefont{T1}{ptm}{m}{n}
\rput(2.5194063,0.355){$Q$}
\usefont{T1}{ptm}{m}{n}
\rput(7.4994063,0.375){$P$}
\psframe[linewidth=0.04,dimen=outer](3.0771875,0.595)(2.0371876,0.125)
\psframe[linewidth=0.04,dimen=outer](4.0771875,0.595)(3.0371876,0.125)
\psframe[linewidth=0.04,dimen=outer](5.0771875,0.595)(4.0371876,0.125)
\psframe[linewidth=0.04,dimen=outer](6.0771875,0.595)(5.0371876,0.125)
\psframe[linewidth=0.04,dimen=outer](7.0771875,0.595)(6.0371876,0.125)
\psframe[linewidth=0.04,dimen=outer](8.077188,0.595)(7.0371876,0.125)
\usefont{T1}{ptm}{m}{n}
\rput(5.5394063,-0.085){$D$}
\usefont{T1}{ptm}{m}{n}
\rput(7.4994063,-0.085){$D$}
\usefont{T1}{ptm}{m}{n}
\rput(4.5394063,-0.085){$D$}
\usefont{T1}{ptm}{m}{n}
\rput(6.5194063,-0.105){$D$}
\usefont{T1}{ptm}{m}{n}
\rput(2.4994063,-0.085){$P$}
\usefont{T1}{ptm}{m}{n}
\rput(3.5194066,-0.085){$Q$}
\psframe[linewidth=0.04,dimen=outer](3.0771875,0.155)(2.0371876,-0.315)
\psframe[linewidth=0.04,dimen=outer](4.0771875,0.155)(3.0371876,-0.315)
\psframe[linewidth=0.04,dimen=outer](5.0771875,0.155)(4.0371876,-0.315)
\psframe[linewidth=0.04,dimen=outer](6.0771875,0.155)(5.0371876,-0.315)
\psframe[linewidth=0.04,dimen=outer](7.0771875,0.155)(6.0371876,-0.315)
\psframe[linewidth=0.04,dimen=outer](8.077188,0.155)(7.0371876,-0.315)
\usefont{T1}{ptm}{m}{n}
\rput(2.5194063,-0.525){$D$}
\usefont{T1}{ptm}{m}{n}
\rput(5.5394063,-0.545){$D$}
\usefont{T1}{ptm}{m}{n}
\rput(6.5394063,-0.545){$D$}
\usefont{T1}{ptm}{m}{n}
\rput(7.4994063,-0.525){$D$}
\psframe[linewidth=0.04,dimen=outer](3.0771875,-0.285)(2.0371876,-0.755)
\psframe[linewidth=0.04,dimen=outer](4.0771875,-0.285)(3.0371876,-0.755)
\psframe[linewidth=0.04,dimen=outer](5.0771875,-0.285)(4.0371876,-0.755)
\psframe[linewidth=0.04,dimen=outer](6.0771875,-0.285)(5.0371876,-0.755)
\psframe[linewidth=0.04,dimen=outer](7.0771875,-0.285)(6.0371876,-0.755)
\psframe[linewidth=0.04,dimen=outer](8.077188,-0.285)(7.0371876,-0.755)
\usefont{T1}{ptm}{m}{n}
\rput(2.5194063,-0.965){$D$}
\usefont{T1}{ptm}{m}{n}
\rput(3.5194063,-0.965){$D$}
\usefont{T1}{ptm}{m}{n}
\rput(6.5194063,-0.965){$D$}
\usefont{T1}{ptm}{m}{n}
\rput(7.4994063,-0.965){$D$}
\psframe[linewidth=0.04,dimen=outer](3.0771875,-0.725)(2.0371876,-1.195)
\psframe[linewidth=0.04,dimen=outer](4.0771875,-0.725)(3.0371876,-1.195)
\psframe[linewidth=0.04,dimen=outer](5.0771875,-0.725)(4.0371876,-1.195)
\psframe[linewidth=0.04,dimen=outer](6.0771875,-0.725)(5.0371876,-1.195)
\psframe[linewidth=0.04,dimen=outer](7.0771875,-0.725)(6.0371876,-1.195)
\psframe[linewidth=0.04,dimen=outer](8.077188,-0.725)(7.0371876,-1.195)
\usefont{T1}{ptm}{m}{n}
\rput(2.5194063,-1.405){$D$}
\usefont{T1}{ptm}{m}{n}
\rput(3.5194063,-1.405){$D$}
\usefont{T1}{ptm}{m}{n}
\rput(4.5394063,-1.425){$D$}
\usefont{T1}{ptm}{m}{n}
\rput(7.5194063,-1.425){$D$}
\psframe[linewidth=0.04,dimen=outer](3.0771875,-1.165)(2.0371876,-1.635)
\psframe[linewidth=0.04,dimen=outer](4.0771875,-1.165)(3.0371876,-1.635)
\psframe[linewidth=0.04,dimen=outer](5.0771875,-1.165)(4.0371876,-1.635)
\psframe[linewidth=0.04,dimen=outer](6.0771875,-1.165)(5.0371876,-1.635)
\psframe[linewidth=0.04,dimen=outer](7.0771875,-1.165)(6.0371876,-1.635)
\psframe[linewidth=0.04,dimen=outer](8.077188,-1.165)(7.0371876,-1.635)
\usefont{T1}{ptm}{m}{n}
\rput(3.5194063,-0.525){$P$}
\usefont{T1}{ptm}{m}{n}
\rput(4.5394063,-0.525){$Q$}
\usefont{T1}{ptm}{m}{n}
\rput(4.5394063,-0.965){$P$}
\usefont{T1}{ptm}{m}{n}
\rput(5.5394063,-0.965){$Q$}
\usefont{T1}{ptm}{m}{n}
\rput(5.5394063,-1.425){$P$}
\usefont{T1}{ptm}{m}{n}
\rput(6.5394063,-1.405){$Q$}
\end{pspicture} 
}
\caption{Rotated RDP array does not form a balanced $2$-parity group ($p=5$)}
\label{fig:rdp-rotate}
\end{figure}

\begin{definition}
\label{def:hpg}
The balanced $2$-parity group obtained from an RDP array code 
as in Example~\ref{ex:bpg-1} is called an (balanced) \emph{RDP $2$-parity group}. 
An \emph{EVENODD $2$-parity group} and an  \emph{RS $2$-parity group} are
defined in the same way.  
\end{definition}
\vskip 10pt 

\begin{lemma}
\label{lem:hpg-workload}
Suppose $\G$ is a balanced RDP/EVENODD/RS $2$-parity group of size $k$. 
Then to reconstruct a missing column of $\G$, one needs to read 
a portion $\frac{k-2}{k-1}$ of the total content of each other column.
In fact, this also holds for every horizontal code that has the same 
reconstruction rule as the RDP code.  
\end{lemma}
\vskip 10pt
\begin{proof} 
Appendix~\ref{appendix-A}.
\end{proof} 
\vskip 10pt 

\subsection{Design of Declustered-Parity Layouts via $3$-Designs}
\label{subsec:dp-via-3-designs}

Recall that the size $k$ of a $2$-parity group $\G$ is its number of columns.
Each column of $\G$ corresponds to a column-unit in a physical disk, which 
is a column of data/parity units.

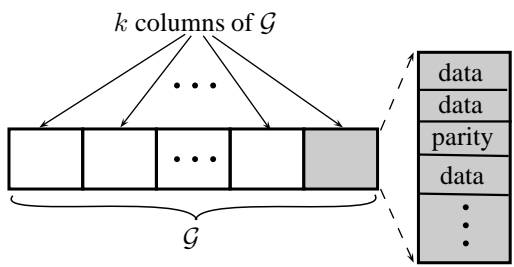
\begin{figure}[h]
\centering
\scalebox{1} 
{
\begin{pspicture}(0,-1.6792186)(6.935,1.7192186)
\definecolor{color1984b}{rgb}{0.8,0.8,0.8}
\psframe[linewidth=0.04,dimen=outer](1.2367188,0.14078125)(0.21671873,-0.698308)
\psbezier[linewidth=0.02](2.714737,-1.0893269)(2.7970016,-0.7300453)(5.326643,-1.0893269)(5.1004148,-0.7548233)
\psbezier[linewidth=0.02](2.7353032,-1.0893269)(2.6119058,-0.71765625)(0.0,-1.0893269)(0.28792664,-0.7672124)
\usefont{T1}{ptm}{m}{n}
\rput(2.6642187,-1.3292187){$\G$}
\psline[linewidth=0.018cm,arrowsize=0.05291667cm 2.0,arrowlength=1.4,arrowinset=0.4]{->}(2.58,1.3807812)(0.64,0.14078125)
\psline[linewidth=0.018cm,arrowsize=0.05291667cm 2.0,arrowlength=1.4,arrowinset=0.4]{->}(2.66,1.3607812)(1.7167187,0.141692)
\psline[linewidth=0.018cm,arrowsize=0.05291667cm 2.0,arrowlength=1.4,arrowinset=0.4]{->}(2.82,1.3607812)(3.7167187,0.101692)
\psline[linewidth=0.018cm,arrowsize=0.05291667cm 2.0,arrowlength=1.4,arrowinset=0.4]{->}(2.94,1.3607812)(4.6767187,0.141692)
\usefont{T1}{ptm}{m}{n}
\rput(2.7142189,1.5307811){$k$ columns of $\G$}
\psline[linewidth=0.018cm,linestyle=dashed,dash=0.16cm 0.16cm,arrowsize=0.05291667cm 2.0,arrowlength=1.4,arrowinset=0.4]{->}(5.1767187,0.101692)(5.616719,1.141692)
\psline[linewidth=0.018cm,linestyle=dashed,dash=0.16cm 0.16cm,arrowsize=0.05291667cm 2.0,arrowlength=1.4,arrowinset=0.4]{->}(5.1767187,-0.678308)(5.616719,-1.658308)
\psframe[linewidth=0.04,dimen=outer,fillstyle=solid,fillcolor=color1984b](6.92,1.1607811)(5.66,-1.6792186)
\usefont{T1}{ptm}{m}{n}
\rput(6.240781,0.87078124){data}
\psline[linewidth=0.03cm](5.68,0.6407812)(6.86,0.6407812)
\psline[linewidth=0.03cm](5.68,0.22078125)(6.92,0.22078125)
\psline[linewidth=0.03cm](5.68,-0.21921875)(6.88,-0.23921874)
\psline[linewidth=0.03cm](5.68,-0.7392188)(6.88,-0.75921875)
\usefont{T1}{ptm}{m}{n}
\rput(6.240781,0.45078123){data}
\usefont{T1}{ptm}{m}{n}
\rput(6.268594,-0.00921875){parity}
\usefont{T1}{ptm}{m}{n}
\rput(6.260781,-0.48921877){data}
\psframe[linewidth=0.04,dimen=outer](2.2167187,0.14078125)(1.1967187,-0.698308)
\psframe[linewidth=0.04,dimen=outer](3.1967187,0.14078125)(2.1767187,-0.698308)
\psframe[linewidth=0.04,dimen=outer](4.1767187,0.14078125)(3.1567187,-0.698308)
\psframe[linewidth=0.04,dimen=outer,fillstyle=solid,fillcolor=color1984b](5.1567187,0.14078136)(4.1367188,-0.698308)
\psline[linewidth=0.08cm,linestyle=dotted,dotsep=0.16cm](2.43,0.6992186)(2.99,0.6992186)
\psline[linewidth=0.08cm,linestyle=dotted,dotsep=0.16cm](2.43,-0.2807814)(2.99,-0.2807814)
\psline[linewidth=0.08cm,linestyle=dotted,dotsep=0.16cm](6.309284,-0.8807823)(6.310716,-1.4407805)
\end{pspicture} 
}
\caption{Simplified layout of a $2$-parity group}
\label{fig:parity-group-simplified}
\end{figure}

The following algorithm extends Algorithm~1 to construct declustered-parity layout for 
two-failure tolerant codes. Compared to Algorithm~1, in the resulting array that 
Algorithm~2 produces, the number of parity units in every column is already balanced
(see Theorem~\ref{thm:main1}). 

\vskip 10pt 
\nin {\bf Algorithm 2}
\begin{itemize}
  \item {\bf Input:} $n$ is the number of physical disks in the array and $k$ is the parity group size.
	\item {\bf Step 1:} Choose a balanced $2$-parity group $\G$ of size $k$ ($\G$ has $k$ columns).
	\item {\bf Step 2:} Choose a $3$-$(n,k,\lam)$ design $\D = (\X,\B)$ for some $\lam$.
	\item {\bf Step 3:} For each block $B_i = \{b_{i,0},\ldots,b_{i,k-1}\} \in \B$, $0 \leq i < |\B|$, create a balanced $2$-parity group $\G_i$ as 
	follows. Firstly, $\G_i$ must have the same data-parity pattern and the same reconstruction rule as $\G$. 
	Secondly, the $k$ columns of $\G_i$ are located on disks with labels $b_{i,0}, \ldots, b_{i,k-1}$. 
	\item {\bf Output:} The $n$-disk array with $|\B|$ parity groups and their layouts according to Step~3.  
\end{itemize}

Note that even though $\G_i$, $0 \leq i < |\B|$, all have the same data-parity pattern of $\G$, on the physical disks, they store
independent sets of data/parity units. The steps in Algorithm~2 are illustrated in the following example. 

\vskip 10pt 
\begin{example}
\label{ex:algo2}
Suppose $\G$ is a balanced $2$-parity group of size four. 
For instance, $\G$ can be obtained from a $2 \times 4$ RDP array ($p = 3$)
using the method described in Example~\ref{ex:bpg-1}. Then the 
simplified layout of $\G$ is as follows (Fig.~\ref{fig:rdp-4}). Each column of $\G$ actually corresponds to a column of 
$24 = 2\times (4 \times 3)$ parity/data units on a physical disk. 
 
\begin{figure}[h]
\centering
\scalebox{1}
{
\begin{pspicture}(0,-0.61)(3.96,0.65)
\usefont{T1}{ptm}{m}{n}
\rput(1.9628124,0.4615625){$\G$}
\psframe[linewidth=0.04,dimen=outer](1.02,0.19156249)(0.0,-0.61)
\psframe[linewidth=0.04,dimen=outer](2.0,0.19156249)(0.98,-0.61)
\psframe[linewidth=0.04,dimen=outer](2.98,0.19156249)(1.96,-0.61)
\psframe[linewidth=0.04,dimen=outer](3.96,0.19156249)(2.94,-0.61)
\end{pspicture} 
}
\caption{A balanced $2$-parity group of size four}
\label{fig:rdp-4}
\end{figure}
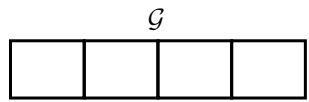

Suppose we have $n = 8$ physical disks. Consider the following $3$-$(8,4,1)$ design
$\D = (\X, \B)$ where
\[
\X = \{0, 1, 2, 3, 4, 5, 6, 7\},
\]
and 
\[
\begin{split}
\B = \Big\{
&\{0, 1, 2, 3\}, \{0, 1, 4, 5\}, \{0,1,6,7\}, \{0,2,4,6\},\\
&\{0,2,5,7\}, \{0,3,4,7\}, \{0,3,5,6\}, \{4,5,6,7\},\\
&\{2,3,6,7\}, \{2,3,4,5\}, \{1,3,5,7\}, \{1,3,4,6\},\\
&\{1,2,5,6\}, \{1,2,4,7\} 
\Big\}.
\end{split}
\]
The resulting array code $\C$ is depicted in Fig.~\ref{fig:resulting-code}. 
There are $14$ $2$-parity groups in $\C$, namely $\G_i$, $0 \leq i < 14$. 
The $2$-parity group $\G_i$ has its columns, labeled by $i$, spread across the disks indexed  
by elements from the block $B_i \in \B$, $0 \leq i < 14$.
For example, as $B_{13} = \{1,2,4,7\}$, the columns of $\G_{13}$, labeled
by $13$, are located on Disk $1$, Disk $2$, Disk $4$, and Disk $7$.   
As each $\G_i$ is a $24 \times 4$ array, $\C$ is actually a $168\times 8$ array 
($168 = 7 \times 24$). 

\begin{figure}[h]
\centering
\scalebox{1}
{
\begin{pspicture}(0,-3.3424733)(8.07,3.3224733)
\definecolor{color1553b}{rgb}{0.8,0.8,0.8}
\psframe[linewidth=0.04,dimen=outer](8.04,2.3224733)(0.02,-3.3224733)
\psline[linewidth=0.04cm](1.0,2.29572)(1.0,-3.29572)
\psline[linewidth=0.04cm](2.0,2.2689667)(2.0,-3.3224733)
\psline[linewidth=0.04cm](3.02,2.2689667)(3.02,-3.3224733)
\psline[linewidth=0.04cm](4.0,2.2689667)(4.0,-3.3224733)
\psline[linewidth=0.04cm](5.0,2.2689667)(5.0,-3.3224733)
\psline[linewidth=0.04cm](6.02,2.29572)(6.02,-3.29572)
\psline[linewidth=0.04cm](7.0,2.2689667)(7.0,-3.3224733)
\psline[linewidth=0.02cm](0.04,1.5224733)(8.04,1.5224733)
\psline[linewidth=0.02cm](0.06,0.7224733)(8.06,0.7224733)
\psline[linewidth=0.02cm](0.04,-0.09752667)(8.04,-0.09752667)
\psline[linewidth=0.02cm](0.04,-0.87752664)(8.04,-0.87752664)
\psline[linewidth=0.02cm](0.04,-1.6775267)(8.04,-1.6775267)
\psline[linewidth=0.02cm](0.04,-2.4775267)(8.04,-2.4775267)
\psframe[linewidth=0.03,dimen=outer,fillstyle=solid,fillcolor=color1553b](8.02,3.3224733)(0.0,2.4624734)
\psline[linewidth=0.04cm](1.0,3.3024733)(1.0,2.5024734)
\psline[linewidth=0.04cm](2.0,3.3024733)(2.0,2.5024734)
\psline[linewidth=0.04cm](3.0,3.3024733)(3.0,2.5024734)
\psline[linewidth=0.04cm](4.0,3.3024733)(4.0,2.5024734)
\psline[linewidth=0.04cm](5.0,3.3024733)(5.0,2.5024734)
\psline[linewidth=0.04cm](6.0,3.3024733)(6.0,2.5024734)
\psline[linewidth=0.04cm](7.0,3.3024733)(7.0,2.5024734)
\usefont{T1}{ptm}{m}{n}
\rput(0.49093747,2.8924732){Disk 0}
\usefont{T1}{ptm}{m}{n}
\rput(1.4996874,2.8924732){Disk 1}
\usefont{T1}{ptm}{m}{n}
\rput(2.50875,2.8924732){Disk 2}
\usefont{T1}{ptm}{m}{n}
\rput(3.51375,2.8924732){Disk 3}
\usefont{T1}{ptm}{m}{n}
\rput(4.51125,2.8924732){Disk 4}
\usefont{T1}{ptm}{m}{n}
\rput(5.5,2.8924732){Disk 5}
\usefont{T1}{ptm}{m}{n}
\rput(6.51,2.8924732){Disk 6}
\usefont{T1}{ptm}{m}{n}
\rput(7.507812,2.8924732){Disk 7}
\usefont{T1}{ptm}{m}{n}
\rput(0.49281248,1.8924733){$0$}
\usefont{T1}{ptm}{m}{n}
\rput(1.4728124,1.8724734){$0$}
\usefont{T1}{ptm}{m}{n}
\rput(2.4728124,1.8724734){$0$}
\usefont{T1}{ptm}{m}{n}
\rput(3.4928124,1.8724734){$0$}
\usefont{T1}{ptm}{m}{n}
\rput(4.4728127,1.8724734){$1$}
\usefont{T1}{ptm}{m}{n}
\rput(5.4728127,1.8724734){$1$}
\usefont{T1}{ptm}{m}{n}
\rput(6.4928126,1.8524734){$2$}
\usefont{T1}{ptm}{m}{n}
\rput(7.4928126,1.8524734){$2$}
\usefont{T1}{ptm}{m}{n}
\rput(0.49281248,1.1324733){$1$}
\usefont{T1}{ptm}{m}{n}
\rput(1.4728124,1.1124734){$1$}
\usefont{T1}{ptm}{m}{n}
\rput(2.4728124,1.1124734){$3$}
\usefont{T1}{ptm}{m}{n}
\rput(3.4928124,1.1124734){$5$}
\usefont{T1}{ptm}{m}{n}
\rput(4.4728127,1.1124734){$3$}
\usefont{T1}{ptm}{m}{n}
\rput(5.4728127,0.29247335){$6$}
\usefont{T1}{ptm}{m}{n}
\rput(6.4928126,1.0924734){$3$}
\usefont{T1}{ptm}{m}{n}
\rput(7.4928126,1.0924734){$4$}
\usefont{T1}{ptm}{m}{n}
\rput(0.49281248,0.31247333){$2$}
\usefont{T1}{ptm}{m}{n}
\rput(1.4728124,0.29247335){$2$}
\usefont{T1}{ptm}{m}{n}
\rput(2.4728124,0.29247335){$4$}
\usefont{T1}{ptm}{m}{n}
\rput(3.4928124,0.29247335){$6$}
\usefont{T1}{ptm}{m}{n}
\rput(4.4728127,0.27247334){$5$}
\usefont{T1}{ptm}{m}{n}
\rput(5.4728127,-0.5275267){$7$}
\usefont{T1}{ptm}{m}{n}
\rput(6.4928126,0.27247334){$6$}
\usefont{T1}{ptm}{m}{n}
\rput(7.4928126,0.27247334){$5$}
\usefont{T1}{ptm}{m}{n}
\rput(0.49281248,-0.48752666){$3$}
\usefont{T1}{ptm}{m}{n}
\rput(2.4728124,-0.5075267){$8$}
\usefont{T1}{ptm}{m}{n}
\rput(3.4928124,-0.5075267){$8$}
\usefont{T1}{ptm}{m}{n}
\rput(4.4728127,-0.5075267){$7$}
\usefont{T1}{ptm}{m}{n}
\rput(5.4728127,-1.3075267){$9$}
\usefont{T1}{ptm}{m}{n}
\rput(6.4928126,-0.5275267){$7$}
\usefont{T1}{ptm}{m}{n}
\rput(7.4928126,-0.5275267){$7$}
\usefont{T1}{ptm}{m}{n}
\rput(0.51281244,-1.3075267){$4$}
\usefont{T1}{ptm}{m}{n}
\rput(1.4828124,-0.5075267){$10$}
\usefont{T1}{ptm}{m}{n}
\rput(2.5128124,-1.3075267){$9$}
\usefont{T1}{ptm}{m}{n}
\rput(3.5128124,-1.3075267){$9$}
\usefont{T1}{ptm}{m}{n}
\rput(4.4928126,-1.3075267){$9$}
\usefont{T1}{ptm}{m}{n}
\rput(5.4628124,-2.1075268){$10$}
\usefont{T1}{ptm}{m}{n}
\rput(6.4528127,-1.3075267){$8$}
\usefont{T1}{ptm}{m}{n}
\rput(7.5128126,-1.3075267){$8$}
\usefont{T1}{ptm}{m}{n}
\rput(0.51281244,-2.1075268){$5$}
\usefont{T1}{ptm}{m}{n}
\rput(1.4828124,-1.3075267){$11$}
\usefont{T1}{ptm}{m}{n}
\rput(2.4828124,-2.1075268){$12$}
\usefont{T1}{ptm}{m}{n}
\rput(3.4828124,-2.1075268){$10$}
\usefont{T1}{ptm}{m}{n}
\rput(4.4628124,-2.1075268){$11$}
\usefont{T1}{ptm}{m}{n}
\rput(5.4628124,-2.9275267){$12$}
\usefont{T1}{ptm}{m}{n}
\rput(6.4828124,-2.1075268){$11$}
\usefont{T1}{ptm}{m}{n}
\rput(7.4828124,-2.1075268){$10$}
\usefont{T1}{ptm}{m}{n}
\rput(0.51281244,-2.8875268){$6$}
\usefont{T1}{ptm}{m}{n}
\rput(1.4828124,-2.1075268){$12$}
\usefont{T1}{ptm}{m}{n}
\rput(2.4828124,-2.9075267){$13$}
\usefont{T1}{ptm}{m}{n}
\rput(3.4828124,-2.9075267){$11$}
\usefont{T1}{ptm}{m}{n}
\rput(4.4628124,-2.9075267){$13$}
\usefont{T1}{ptm}{m}{n}
\rput(6.4828124,-2.9275267){$12$}
\usefont{T1}{ptm}{m}{n}
\rput(7.4828124,-2.9275267){$13$}
\usefont{T1}{ptm}{m}{n}
\rput(5.4728127,1.1124734){$4$}
\usefont{T1}{ptm}{m}{n}
\rput(1.4628124,-2.9075267){$13$}
\end{pspicture} 
}
\caption{The resulting array code $\C$}
\label{fig:resulting-code}
\end{figure}
\end{example}

 \begin{theorem}
\label{thm:main1}
Algorithm~2 produces an array code that satisfies the following properties
\begin{itemize}
	\item[(P1)] it can tolerate at most two simultaneous disk failures; 
	\item[(P2)] when one or two disks fail, the reconstruction workload is 
	evenly distributed to all surviving disks; 
	\item[(P3)] every column of $\C$ has the same number of parity units and data units. 
\end{itemize}
\end{theorem}
\begin{proof}
Appendix~\ref{appendix-B}.  
\end{proof}  
\vskip 10pt 

We now give a high level explanation of how $3$-designs and balanced $2$-parity
groups work well together to produce declustered-parity layouts for two-failure tolerant codes. 

First, let us examine again the application of $2$-designs to one-failure tolerant codes. 
When one disk fails, it is required that all other disks contribute the same amount of data
accesses during the reconstruction process. In other words, we are examining \emph{pairs of 
disks} (one failed, one survived) and want to make sure that all of these pairs have the same
number of related data/parity units (Fig.~\ref{fig:2-disks}). 
(Related units are units that belong to the same parity group). 
On the other hand, in a $2$-design, a similar-looking condition is
applied to \emph{pairs of points}: every pair of points must belong to the same number of blocks. 
That is how the connection between one-failure tolerant codes
and $2$-designs could be established.  
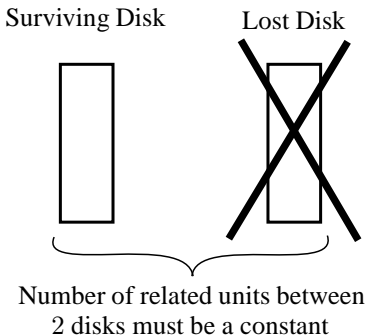
\begin{figure}[h]
\centering
\scalebox{1}
{
\begin{pspicture}(0,-2.20875)(4.8329687,2.20875)
\psframe[linewidth=0.04,dimen=outer](4.1698437,1.45875)(3.429844,-0.68125)
\psframe[linewidth=0.04,dimen=outer](1.4098436,1.45875)(0.6698436,-0.68125)
\psline[linewidth=0.098000005cm](3.0898438,1.75875)(4.6498437,-0.90125)
\psline[linewidth=0.098000005cm](4.509844,1.73875)(2.929844,-0.88125)
\usefont{T1}{ptm}{m}{n}
\rput(3.7979689,2.02875){Lost Disk}
\usefont{T1}{ptm}{m}{n}
\rput(1.0273438,2.02875){Surviving Disk}
\psbezier[linewidth=0.02](2.4387474,-1.3796875)(2.5008354,-0.79968756)(4.4100413,-1.3796875)(4.2393,-0.8396875)
\psbezier[linewidth=0.02](2.4542694,-1.3796875)(2.3611374,-0.7796875)(0.3898438,-1.3796875)(0.6071518,-0.8596875)
\usefont{T1}{ptm}{m}{n}
\rput(2.4360938,-1.63125){Number of related units between}
\usefont{T1}{ptm}{m}{n}
\rput(2.4129686,-2.03125){$2$ disks must be a constant}
\end{pspicture} 
}
\caption{Requirement for any pair of disks}
\label{fig:2-disks}
\end{figure}

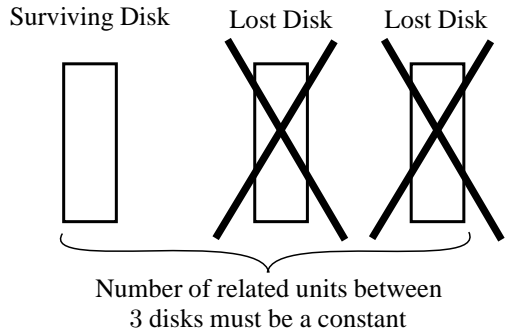
\begin{figure}
\centering
\scalebox{1} 
{
\begin{pspicture}(0,-2.16875)(6.5486875,2.16875)
\psframe[linewidth=0.04,dimen=outer](3.9396873,1.41875)(3.1996872,-0.72125)
\psframe[linewidth=0.04,dimen=outer](1.3996873,1.41875)(0.6596873,-0.72125)
\psline[linewidth=0.098000005cm](2.8596873,1.71875)(4.4196873,-0.94125)
\psline[linewidth=0.098000005cm](4.2796874,1.69875)(2.6996872,-0.92125)
\usefont{T1}{ptm}{m}{n}
\rput(3.5703123,1.98875){Lost Disk}
\usefont{T1}{ptm}{m}{n}
\rput(1.0273438,1.98875){Surviving Disk}
\psbezier[linewidth=0.02](3.387341,-1.3796875)(3.4796941,-0.7996875)(6.3195534,-1.3796875)(6.065582,-0.8396875)
\psbezier[linewidth=0.02](3.4104292,-1.3796875)(3.2718995,-0.7796875)(0.3396873,-1.3796875)(0.6629233,-0.8596875)
\usefont{T1}{ptm}{m}{n}
\rput(3.404375,-1.59125){Number of related units between}
\usefont{T1}{ptm}{m}{n}
\rput(3.3920312,-1.99125){$3$ disks must be a constant}
\psframe[linewidth=0.04,dimen=outer](6.019687,1.41875)(5.2796874,-0.72125)
\psline[linewidth=0.098000005cm](4.9396873,1.71875)(6.499687,-0.94125)
\psline[linewidth=0.098000005cm](6.3596873,1.69875)(4.7796874,-0.92125)
\usefont{T1}{ptm}{m}{n}
\rput(5.6303124,1.98875){Lost Disk}
\end{pspicture} 
}
\caption{Requirement for any group of three disks}
\label{fig:3-disks}
\end{figure}

The problem of designing declustered-parity layouts for two-failure tolerant codes 
also has a similar requirement. It is required that when one or two disks fail, all surviving
disks contribute the same number of data accesses during the reconstruction process. 
Suppose two disks fail. We are in fact examining \emph{groups of three disks} (two failed, one survived)
and want to make sure that all of these groups have the same number of ``related`` data/parity units (Fig.~\ref{fig:3-disks}).
(We use a different meaning here for ``related units". See Appendix~\ref{appendix-B} for more details.)
If we consider a $3$-design, the key property is that every \emph{group of three points} must be 
contained in the same number of blocks. 
At first sight, it is not clear how to translate this condition on points/blocks back 
to the aforementioned condition on disks/groups. 
However, one can do so with the help from some results in Design Theory.  
More details can be found in the proof of Theorem~\ref{thm:main1}
in the Appendix~\ref{appendix-B}. Note also that as a $3$-design is also a $2$-design
(see Corollary~\ref{cr:Stinson1}), uniform workload
for reconstruction of one failed disk is automatically guaranteed. 

The balance of the $2$-parity group used in Algorithm~2 is another key condition
to guarantee the balanced reconstruction workload. In the following
example, it is demonstrated that Algorithm~2 applied to an unbalanced $2$-parity
group does \emph{not} produce a code with this property. 

\begin{figure}[ht]
\centering
\scalebox{1} 
{
\begin{pspicture}(0,-3.3424733)(8.303281,3.3224733)
\definecolor{color1984b}{rgb}{0.8,0.8,0.8}
\psframe[linewidth=0.04,dimen=outer](8.04,2.3224733)(0.02,-3.3224733)
\psline[linewidth=0.04cm](1.0,2.29572)(1.0,-3.29572)
\psline[linewidth=0.04cm](2.0,2.2689667)(2.0,-3.3224733)
\psline[linewidth=0.04cm](3.02,2.2689667)(3.02,-3.3224733)
\psline[linewidth=0.04cm](4.0,2.2689667)(4.0,-3.3224733)
\psline[linewidth=0.04cm](5.0,2.2689667)(5.0,-3.3224733)
\psline[linewidth=0.04cm](6.02,2.29572)(6.02,-3.29572)
\psline[linewidth=0.04cm](7.0,2.2689667)(7.0,-3.3224733)
\psline[linewidth=0.02cm](0.04,1.5224733)(8.04,1.5224733)
\psline[linewidth=0.02cm](0.06,0.7224733)(8.06,0.7224733)
\psline[linewidth=0.02cm](0.04,-0.09752667)(8.04,-0.09752667)
\psline[linewidth=0.02cm](0.04,-0.87752664)(8.04,-0.87752664)
\psline[linewidth=0.02cm](0.04,-1.6775267)(8.04,-1.6775267)
\psline[linewidth=0.02cm](0.04,-2.4775267)(8.04,-2.4775267)
\psframe[linewidth=0.03,dimen=outer,fillstyle=solid,fillcolor=color1984b](8.02,3.3224733)(0.0,2.4624734)
\psline[linewidth=0.04cm](1.0,3.3024733)(1.0,2.5024734)
\psline[linewidth=0.04cm](2.0,3.3024733)(2.0,2.5024734)
\psline[linewidth=0.04cm](3.0,3.3024733)(3.0,2.5024734)
\psline[linewidth=0.04cm](4.0,3.3024733)(4.0,2.5024734)
\psline[linewidth=0.04cm](5.0,3.3024733)(5.0,2.5024734)
\psline[linewidth=0.04cm](6.0,3.3024733)(6.0,2.5024734)
\psline[linewidth=0.04cm](7.0,3.3024733)(7.0,2.5024734)
\usefont{T1}{ptm}{m}{n}
\rput(0.49093747,2.8924732){Disk 0}
\usefont{T1}{ptm}{m}{n}
\rput(1.4796875,2.8924732){Disk 1}
\usefont{T1}{ptm}{m}{n}
\rput(2.50875,2.8924732){Disk 2}
\usefont{T1}{ptm}{m}{n}
\rput(3.49375,2.8924732){Disk 3}
\usefont{T1}{ptm}{m}{n}
\rput(4.51125,2.8924732){Disk 4}
\usefont{T1}{ptm}{m}{n}
\rput(5.5,2.8924732){Disk 5}
\usefont{T1}{ptm}{m}{n}
\rput(6.51,2.8924732){Disk 6}
\usefont{T1}{ptm}{m}{n}
\rput(7.507812,2.8724732){Disk 7}
\usefont{T1}{ptm}{m}{n}
\rput(0.5128125,1.8724734){$D_0$}
\usefont{T1}{ptm}{m}{n}
\rput(1.4928124,1.8724734){$D_0$}
\usefont{T1}{ptm}{m}{n}
\rput(2.4928124,1.8724734){$P_0$}
\usefont{T1}{ptm}{m}{n}
\rput(3.5328124,1.8724734){$Q_0$}
\usefont{T1}{ptm}{m}{n}
\rput(4.4928126,1.8724734){$P_1$}
\usefont{T1}{ptm}{m}{n}
\rput(5.5128126,1.8724734){$Q_1$}
\usefont{T1}{ptm}{m}{n}
\rput(6.5128126,1.8524734){$P_2$}
\usefont{T1}{ptm}{m}{n}
\rput(7.5328126,1.8524734){$Q_2$}
\usefont{T1}{ptm}{m}{n}
\rput(0.5328125,1.1124732){$D_1$}
\usefont{T1}{ptm}{m}{n}
\rput(1.5128125,1.1124734){$D_1$}
\usefont{T1}{ptm}{m}{n}
\rput(2.5128124,1.1124734){$D_3$}
\usefont{T1}{ptm}{m}{n}
\rput(3.5128124,1.1124734){$D_5$}
\usefont{T1}{ptm}{m}{n}
\rput(4.4928126,1.1124734){$P_3$}
\usefont{T1}{ptm}{m}{n}
\rput(5.4928126,0.29247335){$P_6$}
\usefont{T1}{ptm}{m}{n}
\rput(6.5128126,1.0924734){$Q_3$}
\usefont{T1}{ptm}{m}{n}
\rput(7.5328126,1.0924734){$Q_4$}
\usefont{T1}{ptm}{m}{n}
\rput(0.5328125,0.29247332){$D_2$}
\usefont{T1}{ptm}{m}{n}
\rput(1.4928124,0.29247335){$D_2$}
\usefont{T1}{ptm}{m}{n}
\rput(2.5128124,0.29247335){$D_4$}
\usefont{T1}{ptm}{m}{n}
\rput(3.5128124,0.29247335){$D_6$}
\usefont{T1}{ptm}{m}{n}
\rput(4.4928126,0.29247335){$P_5$}
\usefont{T1}{ptm}{m}{n}
\rput(5.5128126,-0.5075267){$D_7$}
\usefont{T1}{ptm}{m}{n}
\rput(6.5328126,0.29247335){$Q_6$}
\usefont{T1}{ptm}{m}{n}
\rput(7.5328126,0.29247335){$Q_5$}
\usefont{T1}{ptm}{m}{n}
\rput(0.5328125,-0.48752666){$D_3$}
\usefont{T1}{ptm}{m}{n}
\rput(2.5128124,-0.5075267){$D_8$}
\usefont{T1}{ptm}{m}{n}
\rput(3.5328124,-0.5075267){$D_8$}
\usefont{T1}{ptm}{m}{n}
\rput(4.5128126,-0.5075267){$D_7$}
\usefont{T1}{ptm}{m}{n}
\rput(5.5128126,-1.3075267){$Q_9$}
\usefont{T1}{ptm}{m}{n}
\rput(6.5128126,-0.5075267){$P_7$}
\usefont{T1}{ptm}{m}{n}
\rput(7.5128126,-0.5075267){$Q_7$}
\usefont{T1}{ptm}{m}{n}
\rput(0.51281244,-1.3075267){$D_4$}
\usefont{T1}{ptm}{m}{n}
\rput(1.5028125,-0.4875267){$D_{10}$}
\usefont{T1}{ptm}{m}{n}
\rput(2.5128124,-1.3075267){$D_9$}
\usefont{T1}{ptm}{m}{n}
\rput(3.5128124,-1.3075267){$D_9$}
\usefont{T1}{ptm}{m}{n}
\rput(4.5128126,-1.3075267){$P_9$}
\usefont{T1}{ptm}{m}{n}
\rput(5.4828124,-2.0875268){$P_{10}$}
\usefont{T1}{ptm}{m}{n}
\rput(6.4928126,-1.3075267){$P_8$}
\usefont{T1}{ptm}{m}{n}
\rput(7.5128126,-1.3075267){$Q_8$}
\usefont{T1}{ptm}{m}{n}
\rput(0.51281244,-2.1075268){$D_5$}
\usefont{T1}{ptm}{m}{n}
\rput(1.4828124,-1.3075267){$D_{11}$}
\usefont{T1}{ptm}{m}{n}
\rput(2.5628126,-2.1075268){$D_{12}$}
\usefont{T1}{ptm}{m}{n}
\rput(3.5028124,-2.1075268){$D_{10}$}
\usefont{T1}{ptm}{m}{n}
\rput(4.5028124,-2.1075268){$P_{11}$}
\usefont{T1}{ptm}{m}{n}
\rput(5.4828124,-2.8875268){$P_{12}$}
\usefont{T1}{ptm}{m}{n}
\rput(6.5228124,-2.1075268){$Q_{11}$}
\usefont{T1}{ptm}{m}{n}
\rput(7.5028124,-2.0875268){$Q_{10}$}
\usefont{T1}{ptm}{m}{n}
\rput(0.51281244,-2.9075267){$D_6$}
\usefont{T1}{ptm}{m}{n}
\rput(1.5228125,-2.1075268){$D_{12}$}
\usefont{T1}{ptm}{m}{n}
\rput(2.5028124,-2.8875268){$D_{13}$}
\usefont{T1}{ptm}{m}{n}
\rput(3.4828124,-2.8875268){$D_{11}$}
\usefont{T1}{ptm}{m}{n}
\rput(4.4828124,-2.8875268){$P_{13}$}
\usefont{T1}{ptm}{m}{n}
\rput(6.5028124,-2.8875268){$Q_{12}$}
\usefont{T1}{ptm}{m}{n}
\rput(7.4828124,-2.8875268){$Q_{13}$}
\usefont{T1}{ptm}{m}{n}
\rput(5.5128126,1.1124734){$P_4$}
\usefont{T1}{ptm}{m}{n}
\rput(1.4628124,-2.9075267){$D_{13}$}
\end{pspicture} 
}
\caption{Unbalanced input leads to unbalanced output}
\label{fig:unbalanced}
\end{figure}

\begin{example}
\label{ex:unbalanced}
Suppose the $3$-design $\D$ in Example~\ref{ex:algo2} and $\G$, an RDP
$2 \times 4$ array, are used in Algorithm~2. Note that $\G$ is an unbalanced
$2$-parity group with the reconstruction rule given in Fig.~\ref{fig:rdp-reconstruction}. 
The layout of the resulting code is depicted in Fig.~\ref{fig:unbalanced}. 

Suppose Disk $0$ and Disk $1$ fail. Let us examine the number of column-units 
on Disk $4$ and $6$, respectively, that need to be accessed for reconstruction of
Disk~$0$ and Disk~$1$. According to the reconstruction rule of each group (Fig.~\ref{fig:rdp-reconstruction}), 
\emph{five} column-units on Disk $4$ must be accessed, whereas only \emph{one}
column-unit on Disk $6$ must be accessed (see Fig.~\ref{fig:accessed-units}). 
Therefore, the workload for reconstruction of the first two disks is not 
uniformly distributed to the surviving disks.   

\begin{figure}[h]
\centering
\begin{tabular}{|c|c|c|c|c|}
\hline
& Disk $0$ & Disk $1$ & Disk $4$ & Disk $6$ \\
\hline
Group $\G_0$ & $D$ & $D$ & X & X \\
Group $\G_1$ & $D$ & $D$ & $\underline{P}$ & X \\
Group $\G_2$ & $D$ & $D$ & X & $\underline{P}$ \\
Group $\G_3$ & $D$ & X & $\underline{P}$ & $Q$ \\
Group $\G_4$ & $D$ & X & X & X \\
Group $\G_5$ & $D$ & X & $\underline{P}$ & X \\
Group $\G_6$ & $D$ & X & X & $Q$ \\
Group $\G_7$ & X & X & $D$ & $P$ \\
Group $\G_8$ & X & X & X & $P$ \\
Group $\G_9$ & X & X & $P$ & X \\
Group $\G_{10}$ & X & $D$ & X & X \\
Group $\G_{11}$ & X & $D$ & $\underline{P}$ & $Q$ \\
Group $\G_{12}$ & X & $D$ & X & $Q$ \\
Group $\G_{13}$ & X & $D$ & $\underline{P}$ & X \\
\hline
\end{tabular}
\caption{Related column-units on Disks $0$, $1$, $4$, and $6$. The underlined entries are those which
must be accessed for reconstruction of Disks $0$ and $1$. An 'X' in a row labeled by Group $\G_i$ and in 
a column labeled by Disk $j$ means that Disk $j$ does not contain any column-unit from $\G_i$.}
\label{fig:accessed-units}
\end{figure}
 
The reason why Algorithm~2 fails to produce a desired array code in the above example
can be explained as follows. Even though the $3$-design spreads out the columns of the $2$-parity
groups evenly among the disks, the columns within each group do not play the same role in the reconstruction
of a lost column. More specifically, the $P$-column and the corresponding $Q$-column
do have different roles in the reconstruction of a $D$-column. Indeed, according to the
reconstruction rule for RDP arrays stated in Fig.~\ref{fig:rdp-reconstruction}, the reconstruction 
of a $D$-column requires the access to the $P$-column, but not to the $Q$-column. 
For example, even though both Disk $4$ and Disk $6$ contain column-units from $\G_3$, 
the column-unit $P_3$ on Disk $4$ must be read, while the column-unit $Q_3$ on Disk $6$
is not read (see Fig.~\ref{fig:accessed-units}). If a balanced $2$-parity group is used instead, 
we will not have this problem, as every column in a balanced $2$-parity group plays the same
role in the reconstruction of a missing column.   
\end{example} 

\subsection{Storage Efficiency and Reconstruction Workload Trade-Off}
\label{subsec:trade-off}

In this subsection we examine the trade-off (of the declustered-parity layout produced by Algorithm~2) between storage efficiency 
and the workload on every disk during the reconstruction of disk failures. 
If an $M \times n$ array code $\C$ contains $x$ parity units and $Mn - x$ data units then we 
say that the number of \emph{disks worth of parity} in $\C$
is $\frac{x}{M}$. The ratio $n - \frac{x}{M}$ is called the number of \emph{disks worth
of data} of $\C$. In other words, $\C$ uses $\frac{x}{M}$ disks to store parities and
$n - \frac{x}{M}$ disks to store data. 

Another attribute of the array code $\C$ produced by Algorithm~2 that 
needs to be examined is the number of rows $M$, or \emph{depth}, of $\C$. 
The depth of $\C$ counts how many units are there in each of its columns. 
An array with fewer rows results in a smaller-size table being stored in the memory and faster (table) look-up. 
Furthermore, a code with a smaller depth provides a better local balance 
(see Schwabe and Sutherland~\cite{SchwabeSutherland1996}). 
The depth of $\C$ depends on $n$, $k$, 
and $\lambda$, as shown in the following theorem. When $n$ and $k$
are fixed, the bigger the index $\lambda$ is, the more rows $\C$ has. 
Therefore, $3$-designs with smaller $\lambda$ are preferred. 
 
\vskip 10pt 
\begin{theorem}
\label{pro:trade-off}
The array code $\C$ produced by Algorithm~2 satisfies the following properties:
\begin{enumerate}
 \item[(P4)] $\C$ has 
\[
M = m\dfrac{\lambda (n-1)(n-2)}{(k-1)(k-2)}
\]
rows, where $m$ is the number of rows in the $2$-parity group $\G$;
	\item[(P5)] $\C$ has $\frac{(k-2)n}{k}$ disks worth of data and $\frac{2n}{k}$ disks worth of parity. 
\end{enumerate}
Moreover, if an RDP/EVENODD/RS $2$-parity group is used in Algorithm~2  
then $\C$ also satisfies the following properties:
\begin{enumerate}
	\item[(P6)] To reconstruct one failed disk, a portion $\frac{k-2}{n-1}$ of the total content of each surviving disk needs to be read; 
	\item[(P7)] To reconstruct two failed disks, a portion $\frac{(k-2)(2n-k-1)}{(n-1)(n-2)}$ of the total content of each surviving disk needs to be read. 
\end{enumerate}
\end{theorem}
\begin{proof} 
Appendix~\ref{appendix-C}. 
\end{proof} 
\vskip 10pt 

When $k = n$, that is, there is no parity declustering involved, Theorem~\ref{pro:trade-off} 
states the familiar facts about an MDS two-failure tolerant array code: $\C$ has $n - 2 = \frac{(n-2)n}{n}$ disks worth of data
and $2 = \frac{2n}{n}$ disks worth of parity; to reconstruct one failed disk, a portion $\frac{n-2}{n-1}$ of the total content of each surviving disk needs to be read; and to reconstruct two failed disks, each surviving disk needs to be read in full ($1 = \frac{(n-2)(2n-n-1)}{(n-1)(n-2)}$). Note that 
the second property does not hold for most of known MDS array codes in 
their original formulations. In fact, it only holds for these codes after some
transformation, such as the one in Example~\ref{ex:bpg-1}, is applied. 

\begin{example}
In this example, we fix the number of disks in the array to be $n = 20$. 
The parity group size $k$ varies from $3$ to $20$. The availability of
a particular $3$-$(n, k, \lambda)$ design can be found in \cite[Part II, 
Table 4.37]{Colbourn2006}. Note that a $t$-design, $t > 3$, is also a $3$-design.
In this table we choose $\lambda$ to be the smallest possible.   

\begin{figure}[h]
\centering
\begin{tabular}{|c|c|c|c|c|c|}
\hline
$k$ & $\lambda$ & $1$ failure & $2$ failures & Parity & 
depth/$m$ \\
\hline
$3$ & 1 & 5.3\% & 10.5\% & 13.3  & 171 \\
$4$ & 1 & 10.5\% & 20.5\% & 10.0  & 57 \\
$5$ & 6 & 15.8\% & 29.8\%  & 8.0 & 171 \\
$6$ & 10 & 21.1\% & 38.6\% & 6.7 & 171 \\
$7$ & 35 & 26.3\% & 46.8\% & 5.7 & 399 \\
$8$ & 14 & 31.6\% & 54.4\% & 5.0 & 114 \\
$9$ & 28 & 36.8\% & 61.4\% & 4.4 & 171 \\
$10$ & 4 & 42.1\% & 67.8\% & 4.0 & 19 \\
$11$ & 55 & 47.4\% & 73.7\% & 3.6 & 209 \\
$12$ & 55 & 52.6\% & 78.9\% & 3.3 & 171 \\
$13$ &286 & 57.9\% & 83.6\% & 3.1 & 741 \\
$14$ & 182 & 63.2\% & 87.7\% & 2.9 & 399 \\
$15$ & 273 & 68.4\% & 91.2\% & 2.7 & 513 \\
$16$ & 140  & 73.7\% & 94.2\% & 2.5 & 228 \\
$17$ & 680 & 78.9\% & 96.5\% & 2.4 & 969 \\
$18$ & 136 & 84.2\% & 98.2\% & 2.2 & 171 \\
$19$ & 17 & 89.5\% & 99.4\% & 2.1 & 19 \\
$20$ & 1 & 94.7\% & 100\% & 2.0 & 1 \\
\hline
\end{tabular}
\caption{Different parity group sizes lead to array codes with different performances ($n = 20$)}
\label{fig:table-n=20}
\end{figure}

The third and fourth columns show the percentage of data/parity units that have to be read
on each surviving disk in order to reconstruct one and two failed disks, respectively.
The fifth column presents the number of parity disks to be used when the corresponding
parity group size $k$ is used. 
The figures in the third, fourth, and fifth columns only 
depend on $n$ and $k$. 
As expected, when $k$ increases, the percentage of units that have to be
accessed for disk recovery increases, and the number of parity disks used decreases. 
Thus, one has to trade the storage efficiency for the reconstruction workload (on each disk): 
increasing storage efficiency, which is good, leads to increasing workload during disk recovery, 
which is bad, and vice versa.  
One extreme is when $k = n$, where there is no parity declustering. 
The array code becomes a normal MDS array code, with two disks worth of parities 
and $100\%$ load on every surviving disk during the reconstruction of two failed disks. 

The figures in the last column are the depths of the resulting
array codes divided by $m$, the depths of the balanced $2$-parity groups $\G$ 
(see Algorithm~2). These figures depend on $n$, $k$, and $\lambda$.   

The ingredient balanced $2$-parity groups $\G$ of size $k$ ($ 3 \leq k \leq 20$) 
can be constructed using the method presented in Example~\ref{ex:bpg-1}. 
This method can be applied to 
an RS code of length $k$ for an arbitrary $k \geq 3$ to obtain a $(k(k-1))\times k$ balanced $2$-parity
group ($m = k(k-1)$).
For an EVENODD code~\cite{BlaumBradyBruckMenon1995}, this method 
produces a $(k(k-1)(k-3)) \times k$ balanced $2$-parity group ($m = k(k-1)(k-3)$), for 
every $k = p + 2$ where $p$ is a prime.  
For an RDP code~\cite{Corbett2004}, this method produces a $(k(k-1)(k-2)) \times k$ 
balanced $2$-parity group ($m = k(k-1)(k-2)$), for every $k = p + 1$ where $p$ is a prime.   
\end{example} 

\vskip 10pt
\begin{remark}
Corbett introduced in his patent \cite{Corbett2008} a method to mix $n/2$ data disks from one array code with $n/2$ data disks from another code to produce an array code that has $n$ data disks. When 
one or two disks fail, the reconstruction workload is distributed evenly to all surviving data disks (but not to all data/parity disks). 
His method actually uses the \emph{complete} $3$-$(n, n/2, \lam)$ design $(\X, \B)$ where all $(n/2)$-subsets of $\X$ are blocks.
In fact, any \emph{self-complementary} $3$-designs would work well with his construction (a design is self-complementary if it satisfies that $B\in \B$ if and only if $\X \setminus B \in \B$). The Hadamard $3$-$(n, n/2, n/4-1)$ design is such a design (see~\cite{AssmusKey1992}). 
Using a Hadamard design results in an array code of only $m(n - 1)$ rows, where $m$ is the 
depth of the original array codes. By contrast, the construction in \cite{Corbett2008} produces an array code of an
extremely large depth $m\binom{n}{n/2}$.  
\end{remark}

\section{Parity Declustering for $(t-1)$-Failure Tolerant Codes via $t$-designs} 
\label{sec:pd-(t-1)-failures}

The generalization of Algorithm~2 to Algorithm~3 below that works for 
$(t-1)$-failure tolerant codes ($t \geq 2$) is straightforward. 

\vskip 10pt
\nin {\bf Algorithm 3}
\begin{itemize}
  \item {\bf Input:} $n$ is the number of physical disks in the array and $k$ is the parity group size.
	\item {\bf Step 1:} Choose a balanced $(t-1)$-parity group $\G$ of size $k$.
	\item {\bf Step 2:} Choose a $t$-$(n,k,\lam)$ design $\D = (\X,\B)$.
	\item {\bf Step 3:} For each block $B_i = \{b_{i,0},\ldots,b_{i,k-1}\} \in \B$, $0 \leq i < |\B|$, 
	create a balanced $(t-1)$-parity group $\G_i$ as follows. 
	Firstly, $\G_i$ must have the same data-parity pattern and the same reconstruction rule as $\G$. 
	Secondly, the $k$ columns of $\G_i$ are located on disks with labels $b_{i,0},\ldots,b_{i,k-1}$. 
	\item {\bf Output:} The $n$-disk array with $|\B|$ parity groups and their layouts according to Step~3.  
\end{itemize}
\vskip 10pt

Relevant $t$-designs can be found in \cite[Part II, Table 4.37]{Colbourn2006} and in the references therein. 
The ingredient balanced $(t-1)$-parity group $\G$ in Algorithm~3 can be constructed by applying the method
in Example~\ref{ex:bpg-1} to any MDS horizontal array code that tolerates $t-1$ disk failures. 
More specifically, suppose that the original array code has $k - t + 1$ data columns ($D$) and $t-1$ parity columns, 
namely $P_i$-columns, $i = 1,\ldots, t-1$. There are $(t-1)!\binom{k}{t-1}$ ways to 
arrange the parity columns of the original array. For each of such arrangements, we obtain a new array. 
By juxtaposing vertically all of these $(t-1)!\binom{k}{t-1}$ arrays, we obtain a balanced 
$(t-1)$-parity group. 
The proof that the above method works for general $t$ is almost the same
as for $t = 3$. For example, for $t = 4$, instead of considering just $r_{DQ}$, $r_{PQ}$, and $r_{QP}$, 
we now need to consider other quantities, such as $r_{P_1P_2}$, $r_{DP_1P_2}$, or $r_{P_1P_2P_3}$. 
They are, in fact, all constants. Therefore, the arguments go the same way as in Example~\ref{ex:bpg-1}. 
We will not provide a detailed proof here. 

Except from the well-known RS codes, some
other known MDS horizontal $(t-1)$-failure tolerant codes ($t > 3$) were studied by 
Blomer {\et}~\cite{Blomer-et-1995}, Blaum {\et} \cite{Blaum-et-1996, Blaum-et-2002}, 
Huang and Xu~\cite{HuangXu2008}.    

\section{Conclusion}
\label{sec:conclusion}

We propose a way to extend the parity declustering technique 
to multiple-failure tolerant array codes based on balanced $(t-1)$-parity groups and $t$-designs ($t \geq 2$). 
Balanced $(t-1)$-parity groups can be obtained from any known horizontal array codes that tolerate up to $t-1$ disk failures.
Besides, $t$-design is a very well-studied combinatorial object in 
the theory of Combinatorial Designs.     
Therefore, one of the advantages of our approach is that we can 
exploit the rich literature from both Erasure Codes theory and 
Combinatorial Designs theory. 

The second advantage of the approach based on $t$-designs is its flexibility. By simply using different $t$-designs in the array code construction, one can obtain a variety of
different trade-offs between storage efficiency and the recovery time.
Note that $\D = (\X, \B)$ where $\B$ consists of all $k$-subset of $\X$ is a $t$-design (called the trivial design) for any $1 \leq t \leq k \leq n$. 
Therefore, for any given number of disks $n$ and any given parity group size $k \leq n$, there always
exists a $t$-$(n, k, \lam)$ design for some $\lambda$. 

One disadvantage of this approach is that sometimes, the 
smallest $t$-design still has an unacceptably large index $\lam$, which 
leads to an impractically deep array code. 
A natural question to ask is whether the depth of the array code, 
in those cases, can be reduced if we relax some requirements on
the array code. A similar question, which is aimed to one-failure tolerant array codes, has already been discussed by Schwabe and Sutherland \cite{SchwabeSutherland1996}. 
Another open question is on the issue of constructing a balanced $(t-1)$-parity group. In this work, we show that \emph{horizontal} array codes can be employed to produce such parity groups. However, the question
of whether \emph{vertical} array codes can also be useful is still open.  

\section{Acknowledgment}
The first author thanks Xing Chaoping for helpful discussions. 

\bibliographystyle{plain}
\bibliography{Parity-Declustering-via-t-designs}
  
\appendices

\section{Proof of Lemma~\ref{lem:hpg-workload}}
\label{appendix-A}
From Example~\ref{ex:bpg-1}, for the recovery of one lost column of $\G$, 
one needs to read
\[
\begin{split}
k(k-1) - r_{DQ} - r_{PQ} &- r_{QP}\\ 
&= k(k-1) - (k - 2) - 1 - 1\\
&= k(k-2)
\end{split}
\]
column-entries in each of the other columns. As each column contains 
$k(k-1)$ column-entries, the portion of content of each column that has
to be accessed is
\[
\dfrac{k(k-2)}{k(k-1)} = \dfrac{k-2}{k-1}. 
\]

\section{Proof of Theorem~\ref{thm:main1}}
\label{appendix-B}

\subsection{Known Results from Design Theory}

The following results from Design Theory are useful in our discussion. 

\vskip 10pt 
\begin{theorem}(\cite[Theorem 9.7]{Stinson2003})
\label{thm:Stinson}
Suppose that $(\X, \B)$ is a $t$-$(n, k, \lambda)$ design. 
Suppose that $Y, Z \subseteq \X$, where $Y \cap Z = \varnothing$, 
$|Y| = i$, $|Z| = j$, and $i + j \leq t$. Then there are exactly
\[
\lambda_i^{(j)} = \dfrac{\lambda\binom{n - i - j}{k - i}}{\binom{n - t}{k - t}}
\] 
blocks in $\B$ that contain all the points in $Y$ and none of the points in $Z$. In particular, 
\[
|\B| = \lam_0^{(0)} = \dfrac{\lam n(n-1)\cdots(n-t+1)}{k(k-1)\cdots(k-t+1)}. 
\]
\end{theorem} 

\begin{corollary}
\label{cr:Stinson0}
Suppose that $(\X, \B)$ is a $3$-$(n, k, \lambda)$ design. 
Then any point $x$ of $\X$ is contained in precisely
\begin{equation}
\label{eq:lam_1}
\lambda_1 = \dfrac{\lambda(n-1)(n-2)}{(k-1)(k-2)}
\end{equation}
blocks.
\end{corollary} 
\begin{proof}
Let $t = 3$, $Y = \{x\}$ and $Z = \varnothing$ and apply Theorem~\ref{thm:Stinson}. 
\end{proof} 

\begin{corollary}
\label{cr:Stinson1}
Suppose that $(\X, \B)$ is a $3$-$(n, k, \lambda)$ design. 
Then any two distinct points $x$ and $y$ in $\X$ are contained in
precisely
\begin{equation}
\label{eq:lam_2}
\lambda_2 =  \dfrac{\lambda(n-2)}{k-2}
\end{equation}
blocks.
\end{corollary} 
\begin{proof}
Let $t = 3$, $Y = \{x,y\}$ and $Z = \varnothing$ and apply Theorem~\ref{thm:Stinson}. 
\end{proof} 

\begin{corollary}
\label{cr:Stinson2}
Suppose that $(\X, \B)$ is a $3$-$(n, k, \lambda)$ design.  
Suppose that $x$, $y$, and $z$ are three distinct points in $\X$. 
Then the number of blocks in $\B$ that contain both $x$ and $y$ 
but not $z$ is
\begin{equation}
\label{eq:lam_2^1}
\lambda_2^{(1)} =  \dfrac{\lambda(n - k)}{k-2}. 
\end{equation}
\end{corollary} 
\begin{proof}
Let $t = 3$, $Y = \{x, y\}$, $Z= \{z\}$, and apply Theorem~\ref{thm:Stinson}.  
\end{proof} 
\vskip 10pt 

Now we are ready to prove Theorem~\ref{thm:main1}. 
Let $\C$ be the array code produced by Algorithm~2. 
Suppose that in $\G$ (and hence in every $\G_i$), to recover one (two) missing column, 
precisely $\tau_1$ ($\tau_2$) entries have to be read from every other column.

\subsection{Proof of $\C$ satisfying (P3)}

First note that due to Corollary~\ref{cr:Stinson0}, each column of $\C$
contains precisely $\lambda_1 = \frac{\lam (n-1)(n-2)}{(k-1)(k-2)}$ column-units.
Therefore, each column of $\C$ contains the same number of units. 
Also, as each column of $\G_i$ (that is, each column-unit of $\C$) contains 
the same number of \emph{parity}
units for all $0 \leq i < |\B|$, each column of $\C$ contains the same
number of \emph{parity} units. Thus $\C$ satisfies (P3). 

\subsection{Proof of $\C$ satisfying (P1)}
 
According to Definition~\ref{def:pg}, each $2$-parity group can recover up to two missing columns. 
Moreover, according to Algorithm~2, no two columns of the same group are located (as column-units) in the same column of $\C$. 
Therefore, $\C$ can tolerate up to two disk failures.
Thus $\C$ satisfies (P1).    

\subsection{Proof of $\C$ satisfying (P2)} 

Suppose Disk $y$ of $\C$ fails. 
Let $x$ be an arbitrary surviving disk of $\C$. 
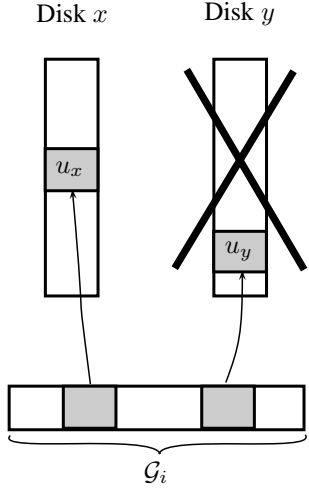
\begin{figure}[h]
\centering
\scalebox{1} 
{
\begin{pspicture}(0,-3.263125)(4.3501587,3.263125)
\definecolor{color1984b}{rgb}{0.8,0.8,0.8}
\psframe[linewidth=0.04,dimen=outer](3.6613774,2.4846876)(2.9213774,-0.7053126)
\psframe[linewidth=0.04,dimen=outer](1.4213772,2.4846876)(0.6813772,-0.7053126)
\psline[linewidth=0.098000005cm](2.5813775,2.3246875)(4.1413774,-0.3353126)
\psline[linewidth=0.098000005cm](4.0013776,2.3046875)(2.4213774,-0.31531262)
\usefont{T1}{ptm}{m}{n}
\rput(3.2863774,3.0746875){Disk $y$}
\usefont{T1}{ptm}{m}{n}
\rput(1.061065,3.0746875){Disk $x$}
\psframe[linewidth=0.04,dimen=outer,fillstyle=solid,fillcolor=color1984b](3.6598148,0.19468741)(2.919815,-0.38531262)
\psframe[linewidth=0.04,dimen=outer,fillstyle=solid,fillcolor=color1984b](1.4198148,1.2946875)(0.6798149,0.69468737)
\usefont{T1}{ptm}{m}{n}
\rput(3.2740335,-0.0953126){$u_y$}
\usefont{T1}{ptm}{m}{n}
\rput(1.0240338,1.0046873){$u_x$}
\psbezier[linewidth=0.02,arrowsize=0.05291667cm 2.0,arrowlength=1.4,arrowinset=0.4]{->}(3.0999997,-1.8368752)(3.2799997,-1.3568752)(3.319815,-1.1769956)(3.319815,-0.3453126)
\psbezier[linewidth=0.02,arrowsize=0.05291667cm 2.0,arrowlength=1.4,arrowinset=0.4]{->}(1.2999997,-1.8568752)(1.0799997,-0.1368751)(1.2629323,-1.5853126)(1.0599997,0.7346874)
\usefont{T1}{ptm}{m}{n}
\rput(2.1740334,-3.0353124){$\G_i$}
\psframe[linewidth=0.04,dimen=outer](4.1598153,-1.8653127)(0.19981492,-2.4853125)
\psbezier[linewidth=0.02](2.2119725,-2.7879233)(2.279002,-2.5132225)(4.3401585,-2.7879233)(4.155827,-2.5321674)
\psbezier[linewidth=0.02](2.22873,-2.7879233)(2.1281855,-2.50375)(0.0,-2.7879233)(0.23460315,-2.5416398)
\psframe[linewidth=0.04,dimen=outer,fillstyle=solid,fillcolor=color1984b](1.6598148,-1.8653127)(0.9198149,-2.4653125)
\psframe[linewidth=0.04,dimen=outer,fillstyle=solid,fillcolor=color1984b](3.4998147,-1.8653127)(2.7598147,-2.4653125)
\end{pspicture} 
}
\caption{One disk fails}
\label{fig:1-disk-fails}
\end{figure}
According to Corollary~\ref{cr:Stinson1}, in points/blocks language, there are $\lambda_2$
blocks in $\B$ that contain both points $x$ and $y$ of $\X$.
Translated to disks/groups language, there are $\lambda_2$
pairs of column-units $(u_x, u_y)$, where $u_x$ is in Disk $x$, $u_y$ is in Disk $y$ and $u_x$ and $u_y$ are from the same $2$-parity group.
For such a pair of column-units $(u_x, u_y)$, in order to recover 
$u_y$, precisely $\tau_1$ units have to be read from $u_x$. Therefore, 
$\lambda_2 \tau_1$ units have to be read from Disk $x$ for the recovery 
of Disk $y$. This number of units is a constant for every pair of Disk $x$
and $y$. Hence, when one disk fails, the reconstruction workload is 
uniformly distributed to all surviving disks. 

Now suppose that Disk $y$ and Disk $z$ of $\C$ fail.
Let $x$ be an arbitrary surviving disk of $\C$. 
A column-unit $u_x$ in Disk $x$ is involved in the reconstruction 
of the two failed disks if and only if one of the following three 
cases holds.
\begin{itemize}
\item {\bf Case 1:} There exist column-units $u_y$ in Disk $y$ and 
	$u_z$ in Disk $z$ so that $u_x$, $u_y$, and $u_z$ all belong to some $2$-parity group $\G_i$. 
	In this case, as $\G_i$ loses two columns, namely $u_y$ and $u_z$, $\tau_2$ units have to be read from
	$u_x$ for the recovery of the lost columns. According to the definition 
	of a $3$-design, there are precisely $\lambda$ such triples $(u_x, u_y, u_z)$.
	\begin{figure}[h]
   \centering
\scalebox{1} 
{
\begin{pspicture}(0,-3.263125)(6.2699428,3.263125)
\definecolor{color1553b}{rgb}{0.8,0.8,0.8}
\psframe[linewidth=0.04,dimen=outer](1.3013774,2.4846876)(0.5613774,-0.70531267)
\psframe[linewidth=0.04,dimen=outer](5.6213775,2.4846876)(4.8813767,-0.70531267)
\psline[linewidth=0.098000005cm](4.5613775,2.3246875)(6.1213775,-0.33531266)
\psline[linewidth=0.098000005cm](5.981377,2.3046875)(4.401377,-0.31531268)
\usefont{T1}{ptm}{m}{n}
\rput(0.94418997,3.0546875){Disk $x$}
\usefont{T1}{ptm}{m}{n}
\rput(5.2488775,3.0746875){Disk $z$}
\psframe[linewidth=0.04,dimen=outer,fillstyle=solid,fillcolor=color1553b](1.2998148,1.4946872)(0.55981493,0.91468734)
\psframe[linewidth=0.04,dimen=outer,fillstyle=solid,fillcolor=color1553b](5.619815,-0.10531265)(4.8798146,-0.70531267)
\usefont{T1}{ptm}{m}{n}
\rput(0.86543983,1.2046875){$u_x$}
\usefont{T1}{ptm}{m}{n}
\rput(5.18544,-0.39531264){$u_z$}
\psbezier[linewidth=0.02,arrowsize=0.05291667cm 2.0,arrowlength=1.4,arrowinset=0.4]{->}(1.3891238,-1.8853128)(1.4599997,-2.1853125)(0.89999974,-1.036875)(0.9463135,0.9031248)
\usefont{T1}{ptm}{m}{n}
\rput(3.1154397,-3.0353124){$\G_i$}
\psframe[linewidth=0.04,dimen=outer](6.079815,-1.8653127)(0.19981492,-2.4853125)
\psbezier[linewidth=0.02](3.1903956,-2.7879233)(3.287074,-2.5132225)(6.259943,-2.7879233)(5.9940767,-2.5321674)
\psbezier[linewidth=0.02](3.2145653,-2.7879233)(3.0695472,-2.50375)(0.0,-2.7879233)(0.3383753,-2.5416398)
\psframe[linewidth=0.04,dimen=outer,fillstyle=solid,fillcolor=color1553b](1.7798148,-1.8653127)(1.0398148,-2.476875)
\psframe[linewidth=0.04,dimen=outer,fillstyle=solid,fillcolor=color1553b](3.4998147,-1.8653127)(2.7598147,-2.476875)
\psframe[linewidth=0.04,dimen=outer](3.3813775,2.4846876)(2.6413774,-0.70531267)
\psline[linewidth=0.098000005cm](2.3013773,2.3246875)(3.8613772,-0.33531266)
\psline[linewidth=0.098000005cm](3.7213774,2.3046875)(2.1413774,-0.31531268)
\usefont{T1}{ptm}{m}{n}
\rput(3.01419,3.0546875){Disk $y$}
\psframe[linewidth=0.04,dimen=outer,fillstyle=solid,fillcolor=color1553b](3.3798149,0.35468736)(2.6398149,-0.22531265)
\usefont{T1}{ptm}{m}{n}
\rput(2.95544,0.06468735){$u_y$}
\psbezier[linewidth=0.02,arrowsize=0.05291667cm 2.0,arrowlength=1.4,arrowinset=0.4]{->}(3.1399996,-1.8568753)(2.9999998,-1.6168753)(2.9999998,-1.1168753)(3.0199997,-0.23687515)
\psframe[linewidth=0.04,dimen=outer,fillstyle=solid,fillcolor=color1553b](5.519815,-1.8653127)(4.7798147,-2.476875)
\psbezier[linewidth=0.02,arrowsize=0.05291667cm 2.0,arrowlength=1.4,arrowinset=0.4]{->}(5.2,-1.876875)(5.16,-1.8568753)(5.2998147,-1.5169959)(5.2998147,-0.6853126)
\end{pspicture} 
}
\caption{Case 1}
\label{fig:case1}
	\end{figure}
\item {\bf Case 2:} There exists a column-unit $u_y$ in Disk $y$ such that that $u_x$ and $u_y$ belong to some $2$-parity group $\G_i$ and moreover, none of the columns of $\G_i$ are located in Disk $z$. In this case, as $\G_i$ loses only one column, namely $u_y$, $\tau_1$ units have to be read from
	$u_x$ for the recovery of this lost column. According to Corollary~\ref{cr:Stinson2}, there are precisely $\lambda_2^{(1)}$ such pairs $(u_x, u_y)$.
	\begin{figure}[h]
	\centering
	\scalebox{1} 
{
\begin{pspicture}(0,-3.263125)(6.2699428,3.263125)
\definecolor{color1553b}{rgb}{0.8,0.8,0.8}
\psframe[linewidth=0.04,dimen=outer](3.6413774,2.4846876)(2.9013774,-0.7053127)
\psframe[linewidth=0.04,dimen=outer](1.2813773,2.4846876)(0.5413768,-0.7053127)
\psline[linewidth=0.098000005cm](2.5613775,2.3246875)(4.1213775,-0.33531272)
\psline[linewidth=0.098000005cm](3.9813774,2.3046875)(2.4013774,-0.31531274)
\usefont{T1}{ptm}{m}{n}
\rput(3.278565,3.0546875){Disk $y$}
\usefont{T1}{ptm}{m}{n}
\rput(0.91325235,3.0746875){Disk $x$}
\psframe[linewidth=0.04,dimen=outer,fillstyle=solid,fillcolor=color1553b](3.6398149,0.19468732)(2.8998148,-0.38531274)
\psframe[linewidth=0.04,dimen=outer,fillstyle=solid,fillcolor=color1553b](1.2798148,1.2946875)(0.53981483,0.69468725)
\usefont{T1}{ptm}{m}{n}
\rput(3.196846,-0.0953127){$u_y$}
\usefont{T1}{ptm}{m}{n}
\rput(0.8268463,1.0046873){$u_x$}
\psbezier[linewidth=0.02,arrowsize=0.05291667cm 2.0,arrowlength=1.4,arrowinset=0.4]{->}(3.0399997,-1.8568753)(3.0399997,-1.676875)(3.4799998,-1.1968751)(3.2599998,-0.3453127)
\psbezier[linewidth=0.02,arrowsize=0.05291667cm 2.0,arrowlength=1.4,arrowinset=0.4]{->}(1.31,-1.876875)(1.2199997,-1.796875)(0.8599997,-1.396875)(0.9198149,0.73468727)
\usefont{T1}{ptm}{m}{n}
\rput(3.0968459,-3.0353124){$\G_i$}
\psframe[linewidth=0.04,dimen=outer](6.079815,-1.8653127)(0.19981492,-2.4853125)
\psbezier[linewidth=0.02](3.1903956,-2.7879233)(3.287074,-2.5132225)(6.259943,-2.7879233)(5.9940767,-2.5321674)
\psbezier[linewidth=0.02](3.2145653,-2.7879233)(3.0695472,-2.50375)(0.0,-2.7879233)(0.3383753,-2.5416398)
\psframe[linewidth=0.04,dimen=outer,fillstyle=solid,fillcolor=color1553b](1.7798148,-1.8653127)(1.0398148,-2.4653125)
\psframe[linewidth=0.04,dimen=outer](5.6413774,2.4846876)(4.9013777,-0.7053127)
\psline[linewidth=0.098000005cm](4.5613775,2.3246875)(6.1213775,-0.33531272)
\psline[linewidth=0.098000005cm](5.9813776,2.3046875)(4.4013777,-0.31531274)
\usefont{T1}{ptm}{m}{n}
\rput(5.2685647,3.0546875){Disk $z$}
\psframe[linewidth=0.04,dimen=outer,fillstyle=solid,fillcolor=color1553b](3.359815,-1.8653127)(2.6198149,-2.4653125)
\end{pspicture} 
}
	\caption{Case 2}
	\label{fig:case2}
	\end{figure}
\item {\bf Case 3:} There exists a column-unit $u_z$ in Disk $z$ such that that $u_x$ and $u_z$ belong to some $2$-parity group $\G_i$ and moreover, none of the columns of $\G_i$ are located in Disk $y$. In this case, as $\G_i$ loses only one column, namely $u_z$, $\tau_1$ units have to be read from
	$u_x$ for the recovery of the lost column. According to Corollary~\ref{cr:Stinson2}, there are precisely $\lambda_2^{(1)}$ such pairs $(u_x, u_z)$.
	\begin{figure}[h]
	\centering
	\scalebox{1} 
{
\begin{pspicture}(0,-3.233125)(6.2903776,3.233125)
\definecolor{color1553b}{rgb}{0.8,0.8,0.8}
\psframe[linewidth=0.04,dimen=outer](3.6813774,2.5146875)(2.9413774,-0.67531264)
\psframe[linewidth=0.04,dimen=outer](1.2013773,2.5146875)(0.46137685,-0.67531264)
\psline[linewidth=0.098000005cm](2.6013775,2.3546875)(4.1613774,-0.3053126)
\psline[linewidth=0.098000005cm](4.0213776,2.3346875)(2.4413774,-0.28531262)
\usefont{T1}{ptm}{m}{n}
\rput(3.289815,3.0446875){Disk $y$}
\usefont{T1}{ptm}{m}{n}
\rput(0.82450235,3.0446875){Disk $x$}
\psframe[linewidth=0.04,dimen=outer,fillstyle=solid,fillcolor=color1553b](1.1998149,1.3246875)(0.45981485,0.7246874)
\usefont{T1}{ptm}{m}{n}
\rput(0.7840338,1.0346874){$u_x$}
\psbezier[linewidth=0.02,arrowsize=0.05291667cm 2.0,arrowlength=1.4,arrowinset=0.4]{->}(1.13,-1.8268751)(1.19,-1.8468751)(0.7399997,-1.0868752)(0.83981484,0.7646875)
\usefont{T1}{ptm}{m}{n}
\rput(3.1540334,-3.0053124){$\G_i$}
\psframe[linewidth=0.04,dimen=outer](6.079815,-1.8353126)(0.19981492,-2.4553125)
\psbezier[linewidth=0.02](3.1903956,-2.7579234)(3.287074,-2.4832225)(6.259943,-2.7579234)(5.9940767,-2.5021675)
\psbezier[linewidth=0.02](3.2145653,-2.7579234)(3.0695472,-2.47375)(0.0,-2.7579234)(0.3383753,-2.5116398)
\psframe[linewidth=0.04,dimen=outer,fillstyle=solid,fillcolor=color1553b](5.399815,-1.8353126)(4.659815,-2.4353125)
\psframe[linewidth=0.04,dimen=outer](5.7613773,2.5146875)(5.0213776,-0.67531264)
\psline[linewidth=0.098000005cm](4.6813774,2.3546875)(6.2413774,-0.3053126)
\psline[linewidth=0.098000005cm](6.1013775,2.3346875)(4.5213776,-0.28531262)
\usefont{T1}{ptm}{m}{n}
\rput(5.3598146,3.0446875){Disk $z$}
\psframe[linewidth=0.04,dimen=outer,fillstyle=solid,fillcolor=color1553b](5.7598147,-0.0753126)(5.019815,-0.6553126)
\usefont{T1}{ptm}{m}{n}
\rput(5.3440337,-0.3653126){$u_z$}
\psbezier[linewidth=0.02,arrowsize=0.05291667cm 2.0,arrowlength=1.4,arrowinset=0.4]{->}(5.0399995,-1.866875)(5.06,-1.3268751)(5.4399996,-1.306875)(5.3799996,-0.6868751)
\psframe[linewidth=0.04,dimen=outer,fillstyle=solid,fillcolor=color1553b](1.5398148,-1.8353126)(0.7998149,-2.4353125)
\end{pspicture} 
}
	\caption{Case 3}
	\label{fig:case3}
	\end{figure} 
\end{itemize}
Therefore, in summary, when Disk $y$ and Disk $z$ fail, the number 
of units to be read from Disk $x$ for the reconstruction is precisely
\[
\lambda \tau_2 + 2\lambda_2^{(1)} \tau_1. 
\]
As this number is a constant for every three distinct disks $x$, $y$, and $z$, 
we conclude that when two disks fail, the reconstruction workload is evenly
distributed across all surviving disks. 	

\section{Proof of Theorem~\ref{pro:trade-off}}
\label{appendix-C}

Suppose the $2$-parity group $\G$ employed in Algorithm~2 has $m$
rows. 
Recall that $\tau_i$, $i = 1, 2$, denotes the number of entries to be read from
every other column when $i$ columns of $\G$ are lost.
If $\G$ is an RDP/EVENODD/RS $2$-parity group then $\tau_1$ and $\tau_2$
can be explicitly computed. 
Indeed, according to Lemma~\ref{lem:hpg-workload}, we have
\begin{equation} 
\label{eq:tau_1}
\tau_1 = m\dfrac{k-2}{k-1}. 
\end{equation} 
When two columns of $\G$ are lost, all $k-2$ other columns have to be read in full for the recovery of the lost columns. Therefore 
\begin{equation}
\label{eq:tau_2}
\tau_2 = m. 
\end{equation}

\subsection{Proof of $\C$ satisfying (P4)}

According to Corollary~\ref{cr:Stinson0}, each column of $\C$ contains precisely $\lambda_1$ column-entries. Moreover, each of these column-entries consists of $m$ entries. Therefore, each column of $\C$ consists of 
\[
M = m\lambda_1 = m\dfrac{\lam (n-1)(n-2)}{(k-1)(k-2)}
\]
entries. 
 
\subsection{Proof of $\C$ satisfying (P5)}

We need to show that $\C$ has $\frac{(k-2)n}{k}$ 
disks worth of data and $\frac{2n}{k}$ disks worth of parity.

There are $|\B|$ $2$-parity balanced groups and each group consists
of $2m$ parity units (see Definition~\ref{def:pg}). Therefore, the total 
number of parity units in $\C$ is $2m|\B|$. Therefore, $\C$ contains
\[
\dfrac{2m|\B|}{M} = \dfrac{2|\B|}{\lam_1} = \dfrac{2\frac{\lam n(n-1)(n-2)}{k(k-1)(k-2)}}{\frac{\lam (n-1)(n-2)}{(k-1)(k-2)}} = \dfrac{2n}{k}
\]
disks worth of parity. We deduce that $\C$ contains
\[
n - \dfrac{2n}{k} = \dfrac{(k-2)n}{k}.
\]
disks worth of data. 

\subsection{Proof of $\C$ satisfying (P6)}

We need to prove that if $\G$ is an RDP/EVENODD/RS $2$-parity group then in order to reconstruct one failed disk, a portion $\frac{k-2}{n-1}$ of the total content of each surviving disk needs to be read.

Suppose one column of $\C$ is lost. According to Appendix~\ref{appendix-B}, $\lam_2 \tau_1$ entries must be read from each other column
for the reconstruction of the missing column. Since each column of $\C$
consists of $M$ entries, a portion
\[
\dfrac{\lam_2 \tau_1}{M} = \dfrac{\frac{\lam(n-2)}{k-2}m\frac{k-2}{k-1}}{\frac{\lam(n-1)(n-2)}{(k-1)(k-2)}m} = \dfrac{k-2}{n-1}
\]
of the total content of each surviving disk must be read. 

\subsection{Proof of $\C$ satisfying (P7)}

We need to show that if $\G$ is an RDP/EVENODD/RS $2$-parity group then in order to reconstruct two failed disks, a portion $\frac{(k-2)(2n-k-1)}{(n-1)(n-2)}$ of the total content of each surviving disk needs to be read. 

Suppose two columns of $\C$ are lost. According to Appendix~\ref{appendix-B}, $\lam \tau_2 + 2\lam_2^{(1)} \tau_1$ entries must be read from each other column for the reconstruction of the two missing columns.
Thus, a portion
\begin{equation}
\label{eq:e1} 
\begin{split}
\dfrac{\lam \tau_2 + 2\lam_2^{(1)} \tau_1}{M}
\end{split}
\end{equation} 
of the total content of each surviving column needs to be read for the recovery of two columns of $\C$. 
Substituting (\ref{eq:tau_1}), (\ref{eq:tau_2}), (\ref{eq:lam_1}), and 
(\ref{eq:lam_2^1}) into (\ref{eq:e1}), the ratio in this equation can be simplified to 
\[
\dfrac{(k-2)(2n-k-1)}{(n-1)(n-2)}. 
\]

\end{document}